\documentclass{article}

\usepackage{fullpage}
\usepackage{amsthm,amssymb,amsmath}  
\usepackage{comment,xspace,enumerate}
\usepackage[capitalise]{cleveref}
\usepackage[utf8]{inputenc}
\usepackage{thmtools}
\usepackage{thm-restate}
\usepackage{authblk}

  \theoremstyle{plain}
  \newtheorem{theorem}{Theorem}[section]
  \newtheorem{lemma}[theorem]{Lemma}  
  \newtheorem{corollary}[theorem]{Corollary}  
  \newtheorem{fact}[theorem]{Fact}
  \newtheorem{observation}[theorem]{Observation}
  \theoremstyle{definition}
  \newtheorem{definition}[theorem]{Definition}
  \newtheorem{example}[definition]{Example}
  \newtheorem{remark}[definition]{Remark}
 
  \newtheorem*{claim}{Claim}

  \title{Pattern Matching and Consensus Problems on Weighted Sequences and Profiles\footnote{Work supported by the Polish Ministry of Science and Higher Education under the `Iuventus Plus' program in 2015-2016 grant no 0392/IP3/2015/73.}}

\author[1]{Tomasz Kociumaka}
\author[2]{Solon P. Pissis}
\author[1,2]{Jakub Radoszewski\footnote{The author is a Newton International Fellow.}$^{,}$}

\affil[1]{Institute of Informatics, University of Warsaw, Warsaw, Poland\\
    \texttt{[kociumaka,jrad]@mimuw.edu.pl}}
\affil[2]{Department of Informatics, King's College London, London, UK\\
    \texttt{solon.pissis@kcl.ac.uk}
}

\date{\vspace{-5ex}}

  \usepackage{amsmath,amsfonts}
  \usepackage{amsthm}
  \usepackage{tikz}
  \usepackage[T1]{fontenc}
  \usepackage[utf8]{inputenc}
  \usepackage{comment}
  \usepackage{multirow}
  \usepackage{todonotes}
  \usepackage[noend, noline, ruled]{algorithm2e}
  \usepackage[capitalise]{cleveref}
      
  \usetikzlibrary{decorations.pathreplacing,calc}

\newsavebox{\mybox}
\newenvironment{dsproblem}[1]
{\begin{center}\begin{lrbox}{\mybox}\begin{minipage}{0.96\columnwidth}#1 \textsc{Problem}\\}
{\end{minipage}\end{lrbox}\fbox{\usebox{\mybox}}\end{center}}

  \newcommand{\defdsproblem}[3]{
  \begin{dsproblem}{#1}
\textbf{Input:} #2\\
\textbf{Output:} #3
  \end{dsproblem}
  }

  \newcommand{\defdsproblempar}[4]{
  \begin{dsproblem}{#1}
\textbf{Input:} #2\\
\textbf{Question:} #3\\
\textbf{Parameters:} #4
  \end{dsproblem}
  }

  \newcommand{\defdsproblemoutpar}[4]{
  \begin{dsproblem}{#1}
\textbf{Input:} #2\\
\textbf{Output:} #3\\
\textbf{Parameters:} #4
  \end{dsproblem}
  }

  \newcommand{\MK}{\textsc{Multichoice Knapsack}\xspace}
  \newcommand{\PM}{\textsc{Profile Matching}\xspace}
  \newcommand{\WPM}{\textsc{Weighted Pattern Matching}\xspace}
  \newcommand{\GWPM}{\textsc{GWPM}\xspace}
  \newcommand{\GWPMFull}{\textsc{General Weighted Pattern Matching}\xspace}
  \newcommand{\WC}{\textsc{Weighted Consensus}\xspace}
  \newcommand{\SDWCFull}{\textsc{Short Dissimilar Weighted Consensus}\xspace}
  \newcommand{\SDWC}{\textsc{SDWC}\xspace}
  \newcommand{\PC}{\textsc{Profile Consensus}\xspace}
  \newcommand{\Knapsack}{\textsc{Knapsack}\xspace}
  \newcommand{\SubsetSum}{\textsc{Subset Sum}\xspace}
  \newcommand{\Sum}{\textsc{Sum}\xspace}

  \newcommand{\floor}[1]{\left\lfloor #1 \right\rfloor}
  \newcommand{\ceil}[1]{\left\lceil #1 \right\rceil}
  \newcommand{\Oh}{\mathcal{O}}
  \newcommand{\Ohtilde}{\tilde{\mathcal{O}}}
    \newcommand{\Ohstar}{\mathcal{O}^*}
 
  \newcommand{\sub}{\subseteq}

  \newcommand{\Occ}{\mathit{Occ}}
  \newcommand{\lcp}{\mathit{lcp}}

  \renewcommand{\H}{\mathcal{H}}
  \newcommand{\C}{\mathcal{C}}
  \newcommand{\B}{\mathcal{B}}
  
  \renewcommand{\L}{\mathcal{L}}
  \newcommand{\R}{\mathcal{R}}
  \renewcommand{\P}{\mathcal{P}}
  \newcommand{\Ls}{\mathcal{L}^*}
  \newcommand{\Rs}{\mathcal{R}^*}

  \newcommand{\fr}{\ensuremath{\frac1z}}
  \newcommand{\frsq}{$\frac{1}{\sqrt{z}}$}
  \newcommand{\match}{\approx_{\frac1z}}
  
  \newcommand{\mayqed}{}
  \newcommand{\mmid}{\mathrm{mid}}
  
  \DeclareMathOperator*{\rank}{rank}
  
  \newcommand{\HammingDistance}{d_H}
  \newcommand{\Score}{\mathrm{Score}}
  \newcommand{\NumStrings}{\mathrm{NumStrings}}

\begin{document}
  \maketitle
\begin{abstract}
  We study pattern matching problems on two major representations of uncertain sequences used in molecular biology:
  weighted sequences (also known as position weight matrices, PWM) and profiles (i.e., scoring matrices).
  In the simple version, in which only the pattern or only the text is uncertain, 
  we obtain efficient algorithms with theoretically-provable running times using a variation of the lookahead scoring technique.
  We also consider a general variant of the pattern matching problems in which both the pattern and the text are uncertain.
  Central to our solution is a special case where the sequences have equal length, called the consensus problem.
  We propose algorithms for the consensus problem parameterized by the number of strings that match one of the sequences.
  As our basic approach, a careful adaptation of the classic meet-in-the-middle algorithm for the knapsack problem is used.
  On the lower bound side, we prove that our dependence on the parameter is optimal up to lower-order terms
  conditioned on the optimality of the original algorithm for the knapsack problem.
\end{abstract}

  \section{Introduction}
  We study two well-known representations of uncertain texts: \emph{weighted sequences} and \emph{profiles}.
  A \emph{weighted sequence} (also known as uncertain sequence or position weight matrix, PWM)
  for every position and every letter of the alphabet specifies the probability of occurrence of this letter at this position;
  see \cref{table:weighted_sequence} for an example.
  A weighted sequence represents many different strings, each with the probability
  of occurrence equal to the product of probabilities of its letters at subsequent positions of the weighted sequence.
  Usually a threshold \fr\ is specified, and one considers only strings that match the weighted sequence with probability at least \fr.
  A \emph{scoring matrix} (or a profile) of length $m$ is an $m \times \sigma$ matrix.
  The \emph{score} of a string of length $m$ is the sum of scores in the scoring matrix of the subsequent
  letters of the string at the respective positions.
  A string is said to match a scoring matrix if its matching score is above a specified threshold $Z$.

      \begin{figure}[htpb]
      \renewcommand*{\arraystretch}{1.2}
      \begin{center}
      \begin{tabular}{|c|c|c|c|}
        \hline
        $X[1]$ & $X[2]$ & $X[3]$ & $X[4]$ \\
        \hline
        \ $\pi_1^{(X)}(\mathtt{a})=1/2$ \ & \ $\pi_2^{(X)}(\mathtt{a})=1$ \ & \ $\pi_3^{(X)}(\mathtt{a})=3/4$ \ & \ $\pi_4^{(X)}(\mathtt{a})=0$ \ \\
        \ $\pi_1^{(X)}(\mathtt{b})=1/2$ \ & \ $\pi_2^{(X)}(\mathtt{b})=0$ \ & \ $\pi_3^{(X)}(\mathtt{b})=1/4$ \ & \ $\pi_4^{(X)}(\mathtt{b})=1$ \ \\
        \hline
      \end{tabular}
      \end{center}
      \caption{A weighted sequence $X$ of length 4 over the alphabet $\Sigma=\{\mathtt{a},\mathtt{b}\}$}\label{table:weighted_sequence}
      \end{figure}

  \subparagraph*{\WPM and \PM}
  First of all, we study the standard variants of pattern matching problems on weighted sequences and profiles,
  in which only the pattern or the text is an uncertain sequence.
  In the best-known formulation of the \WPM problem, we are given a weighted sequence of length $n$, called a text,
  a solid (standard) string of length $m$, called a pattern, both over an alphabet of size $\sigma$,
  and a \emph{threshold probability} \fr.
  We are asked to find all positions in the text where the fragment
  of length $m$ represents the pattern with probability at least \fr.
  Each such position is called an \emph{occurrence} of the pattern in the text;
  we also say that the fragment of the text and the pattern \emph{match}.
  The \WPM problem can be solved in $\Oh(\sigma n \log m)$ time via the Fast Fourier Transform~\cite{KCL_publication}.
  In a more general indexing variant of the problem, considered in
  \cite{amir_weighted_property_matching_j,costas_weighted_suffix_tree_j}, one can preprocess a weighted text
  in $\Oh(n z^2 \log z)$ time to report all $occ$ occurrences of a given solid pattern of length $m$ in $\Oh(m+occ)$ time.
  (A similar indexing data structure, which assumes $z = \Oh(1)$, was presented in~\cite{DBLP:conf/edbt/BiswasPTS16}.)
  Very recently, the index construction time was reduced to $\Oh(nz)$ for constant-sized alphabets \cite{CPM2016}.

  In the classic \PM problem, the pattern is an $m \times \sigma$ profile, the text is a solid string of length $n$, and
  our task is to find all positions in the text where the fragment of length $m$ has a score above a specified
  threshold $Z$.
  A naive approach to the \PM problem works in $\Oh(nm+m\sigma)$ time.
  A broad spectrum of heuristics improving this algorithm in practice is known; for a survey see~\cite{DBLP:journals/tcs/PizziU08}.
  One of the principal techniques, coming in different flavours, is \emph{lookahead scoring} that consists in checking if a partial match
  could possibly be completed by the following highest scoring letters in the scoring matrix and, if not, pruning the
  naive search.
  The \PM problem can also be solved in $\Oh(\sigma n \log m)$ time via the Fast Fourier Transform~\cite{DBLP:journals/jcb/RajasekaranJS02}.
  
  \subparagraph*{\WC and \PC}
  As our most involved contribution, we study a general variant of pattern matching on weighted sequences
  and the consensus problems on uncertain sequences, which are closely related to the \MK problem.
  In the \WC problem, given two weighted sequences of the same length, we are to check if there is
  a string that matches each of them with probability at least \fr.
  A routine to compare user-entered weighted sequences with existing weighted sequences in the database is used,
  e.g., in JASPAR, a well-known database of PWMs \cite{JASPAR}.
  In the \GWPMFull (\GWPM) problem, both the pattern and the text are weighted.
  In the most common definition of the problem (see \cite{DBLP:conf/cwords/BartonP15,costas_weighted_suffix_tree_j}),
  we are to find all fragments of the text that give a positive answer to the \WC problem with the pattern.
  The authors of~\cite{DBLP:conf/cwords/BartonP15} proposed an algorithm for the \GWPM problem based on the weighted prefix table
  that works in $\Oh(n z^2 \log z + n\sigma)$ time.

  In an analogous way to the \WC problem, we define the \PC problem.
  Here we are to check for the existence of a string that matches both the scoring matrices above threshold $Z$.
  The \PC problem is actually a special case of the well-known (especially in practice) \MK problem
  (also known as the \textsc{Multiple Choice Knapsack} problem).
  In this problem, we are given $n$ classes $C_1,\ldots,C_n$ of at most $\lambda$ items each---$N$ items in total---each item $c$ characterized by a value $v(c)$ and a weight $v(c)$.
  The goal is to select one item from each class so that the sums of values and of weights of the items are
  below two specified thresholds, $V$ and $W$.
  (In the more intuitive formulation of the problem, we require the sum of values to be \emph{above} a specified threshold,
  but here we consider an equivalent variant in which both parameters are symmetric.)
  The \MK problem is widely used in practice, but most research concerns approximation or heuristic solutions;
  see \cite{DBLP:books/daglib/0010031} and references therein.
  As far as exact solutions are concerned, the classic meet-in-the middle approach by Horowitz and Sahni~\cite{DBLP:journals/jacm/HorowitzS74},
  originally designed for the (binary) \Knapsack problem, immediately generalizes to 
  an $\Oh^*(\lambda^{\lceil{\frac{n}{2}\rceil}})$-time\footnote{The $\Oh^*$ notation suppresses factors polynomial with respect to the instance size (encoded in binary). } solution for \MK. 
  

  Several important problems can be expressed as special cases of the \MK problem using folklore reductions (see~\cite{DBLP:books/daglib/0010031}).
  This includes the \SubsetSum problem, which for a set of $n$ integers asks whether there is a subset summing up to a given integer $Q$,
  and the $k$-\Sum problem which, for $k=\Oh(1)$ classes of $\lambda$ integers, asks to choose one element from each
  class so that the selected integers sum up to zero. 
  These reductions give immediate hardness results for the \MK problem, and they can be adjusted to yield the same consequences for \PC.
  For the \SubsetSum problem, as shown in \cite{DBLP:conf/mfcs/EtscheidKMR15,DBLP:books/daglib/0069796}, the existence of an $\Ohstar(2^{\varepsilon n})$-time solution for every $\varepsilon > 0$
  would violate the Exponential Time Hypothesis (ETH)~\cite{DBLP:journals/jcss/ImpagliazzoP01,ETHsurvey}.
  Moreover, the $\Oh^*(2^{n/2})$ running time, achieved in \cite{DBLP:journals/jacm/HorowitzS74}, has not been improved yet despite much effort.
  The 3-\Sum conjecture \cite{DBLP:journals/comgeo/GajentaanO95} and the more general $k$-\Sum conjecture state that the 3-\Sum and $k$-\Sum problems cannot be solved in
  $\Oh(\lambda^{2-\varepsilon})$ time and $\Oh(\lambda^{\ceil{\frac{k}{2}}(1-\varepsilon)})$ time, respectively, for any $\varepsilon>0$.


  \subparagraph*{Our Results}
  As the first result, we show how the lookahead scoring technique combined with a data structure
  for answering longest common prefix queries in a string can be applied to obtain simple and efficient
  algorithms for the standard pattern matching problems on uncertain sequences.
  For a weighted sequence, by $R$ we denote the size of its list representation, and by $\lambda$ the
  maximal number of letters with score at least $\frac{1}{z}$ at a single position (thus $\lambda \le \min(\sigma,z)$).
  In the \PM problem, we set $M$ as the number of strings that match the scoring matrix with score above $Z$.
  In general $M \le \sigma^m$, however, we may assume that for practical data this number is actually much smaller.
  We obtain the following running times:
  \begin{itemize}
    \item $\Oh(m\sigma+n \log M)$ for \PM;
    \item $\Oh(R\log^2\log \lambda+n \log z)$ deterministic and $\Oh(R+n \log z)$ randomized (Las Vegas, failure with probability $R^{-c}$
    for any given constant $c$) for \WPM.
  \end{itemize}
  
  The more complex part of our study is related to the consensus problems and to the \GWPM problem.
  Instead of considering \PC, we study the more general \MK.
  We introduce parameters based on the number of solutions with \emph{feasible} weight or value:
  $A_V = |\{(c_1,\ldots,c_n)\,:\,c_i \in C_i\mbox{ for all }i=1,\ldots,n,\,\sum_i v(c_i) \le V\}|$, that is,
  the number of choices of one element from each class that satisfy the value threshold;
  $A_W = |\{(c_1,\ldots,c_n)\,:\,c_i \in C_i\mbox{ for all }i=1,\ldots,n,\,\sum_i w(c_i) \le W\}|$;
  $A = \max(A_V,A_W)$, and $a=\min(A_V,A_W)$.
  We obtain algorithms with the following complexities:
  \begin{itemize}
    \item $\Oh(N+\sqrt{a\lambda} \log A)$ for \MK;
    \item $\Oh(R+\sqrt{z \lambda} (\log \log z+\log \lambda))$ for \WC and $\Oh(n\sqrt{z \lambda} (\log \log z+\log \lambda))$ for \GWPMFull.
  \end{itemize}
  
  Since $a \le A \le \lambda^n$, our running time for \MK in the worst case matches (up to lower order terms) the time complexities of the fastest known solutions
  for both \SubsetSum (also binary \Knapsack) and 3-\Sum. 
  The main novel part of our algorithm for \MK is an appropriate (yet intuitive) notion of ranks of partial solutions.
  We also provide a simple reduction from \MK to \WC, which lets us transfer the negative results to the \GWPM problem.
  \begin{itemize}
    \item The existence of an $\Oh^*(z^{\varepsilon})$-time solution for \WC for every $\varepsilon>0$  would violate the Exponential Time Hypothesis.
    \item For every $\varepsilon>0$, an $\Oh^*(z^{0.5-\varepsilon})$-time solution for \WC would imply an $\Oh^*(2^{(0.5 -\varepsilon)n})$-time algorithm for \SubsetSum.
    \item For every $\varepsilon>0$, an $\Ohtilde(R+z^{0.5}\lambda^{0.5-\varepsilon})$-time\footnote{
      The $\Ohtilde$ notation ignores factors polylogarithmic with respect to the input size.
    } solution for \WC would imply an $\Ohtilde(\lambda^{2-\varepsilon})$-time algorithm for 3-\Sum.
  \end{itemize}
  For the higher-order terms our complexities match the conditional lower bounds;
  therefore, we put significant effort to keep the lower order terms of the complexities
  as small as possible.

  \subparagraph*{Model of Computations}
  For problems on weighted sequences, we assume the word RAM model with word size $w = \Omega(\log n + \log z)$ and $\sigma = n^{\Oh(1)}$.
  We consider the log-probability model of representations of weighted sequences, that is, we assume that
  probabilities in the weighted sequences and the threshold probability \fr\ are all of the form $c^{\frac{p}{2^{dw}}}$,
  where $c$ and $d$ are constants and $p$ is an integer that fits in a constant number of machine words.
  Additionally, the probability 0 has a special representation.
  The only operations on probabilities in our algorithms are multiplications and divisions, which can be
  performed exactly in $\Oh(1)$ time in this model.
  Our solutions to the \MK problem only assume the word RAM model with word size $w=\Omega(\log S+\log a)$,
  where $S$ is the sum of integers in the input instance; this does not affect the $\Oh^*$ running time.

  \subparagraph*{Structure of the Paper}
  We start with Preliminaries, where we formally introduce the problems and the main notions used throughout the paper.
  The following three sections describe our algorithms: in \cref{sec:EWPM} for \PM and \WPM;
  in \cref{sec:MK} for \PC; and in \cref{sec:GWPMReduction} for \WC and \GWPMFull.
  A tailor-made, yet more efficient algorithm for \GWPMFull is presented in \cref{app:SDWC}.
  We conclude with \cref{app:fast}, where we introduce faster algorithms and matching lower bounds
  for \MK and \GWPM in the case that $\lambda$ is large.

  \section{Preliminaries}\label{sec:Preliminaries}
    Let $\Sigma=\{s_1,s_2,\ldots,s_{\sigma}\}$ be an alphabet of size $\sigma$.
    A \emph{string} $S$ over $\Sigma$ is a finite sequence of letters from $\Sigma$.
    We denote the length of $S$ by $|S|$ and, for $1 \le i \le |S|$,
    the $i$-th letter of $S$ by $S[i]$.
    By $S[i..j]$ we denote the string $S[i] \ldots S[j]$ called a \emph{factor} of $S$
    (if $i>j$, then the factor is an empty string).
    A factor is called a \emph{prefix} if $i=1$ and a \emph{suffix} if $j=|S|$.
    For two strings $S$ and $T$, we denote their concatenation by $S \cdot T$ ($ST$ in short).

    For a string $S$ of length $n$, by $\lcp(i,j)$ we denote the length of the longest common prefix of factors $S[i..n]$ and $S[j..n]$.
    The following fact specifies a known efficient data structure answering such queries.
    It consists of the suffix array with its inverse, the LCP table and a data structure for range minimum queries on
    the LCP table; see \cite{AlgorithmsOnStrings} for details.

    \begin{fact}\label{fct:ver}
      Let $S$ be a string of length $n$ over alphabet of size $\sigma = n^{\Oh(1)}$.
      After $\Oh(n)$-time preprocessing, given indices $i$ and $j$ ($1 \le i,j \le n$) one can compute $\lcp(i,j)$ in $\Oh(1)$ time.
    \end{fact}

    The \emph{Hamming distance} between two strings $X$ and $Y$ of the same length,
    denoted by $\HammingDistance(X,Y)$, is the number of positions where the strings have different letters.

    \subsection{Profiles}
    In the \PM problem, we consider a \emph{scoring matrix} (a profile) $P$ of size $m \times \sigma$.
    For $i \in \{1,\ldots,m\}$ and $j \in \{1,\ldots,\sigma\}$, we denote the integer score of the letter $s_j$
    at the position~$i$ by $P[i,s_j]$.
    The \emph{matching score} of a string $S$ of length $m$ with the matrix $P$ is
    $$\Score(S,P) = \sum_{i=1}^m P[i,S[i]].$$
    If $\Score(S,P) \ge Z$ for an integer \emph{threshold} $Z$, then we say that the string $S$ \emph{matches the matrix $P$
    with threshold $Z$}.
    We denote the number of strings $S$ that math $P$ with threshold $Z$ by $\NumStrings_Z(P)$.
    
    For a string $T$ and a scoring matrix $P$, we say that $P$ \emph{occurs in $T$ at position $i$
    with threshold $Z$} if $T[i..i+m-1]$ matches $P$ with threshold $Z$.
    We denote the set of all positions where $P$ occurs in $T$ by $\Occ_Z(P,T)$.
    These notions let us define the \PM problem:

    \defdsproblemoutpar{\PM}{
      A string $T$ of length $n$, a scoring matrix $P$ of size $m \times \sigma$, and
      a threshold $Z$.
    }{
      The set $\Occ_Z(P,T)$.
    }{
      $M = \NumStrings_Z(P)$.
    }

    \subsection{Weighted Sequences}
    A \emph{weighted sequence} $X=X[1] \ldots X[n]$ of length $|X|=n$ over alphabet $\Sigma=\{s_1,s_2,\ldots,s_{\sigma}\}$
    is a sequence of sets of pairs of the form:
    $$X[i] = \{(s_j,\ \pi^{(X)}_i(s_j))\ :\ j \in \{1,2,\ldots,\sigma\}\}.$$
    Here, $\pi_i^{(X)}(s_j)$ is the occurrence probability of the letter $s_j$ at the position $i \in \{1,\ldots,n\}$.
    These values are non-negative and sum up to 1 for a given $i$.
 
    For all our algorithms, it is sufficient that the probabilities sum up to \emph{at most} 1 for each position.
    Also, the algorithms sometimes produce auxiliary weighted sequences with sum of probabilities being smaller than 1 on some positions.

    We denote the maximum number of letters occurring at a single position of the weighted sequence
    (with non-zero probability) by $\lambda$
    and the total size of the representation of a weighted sequence by $R$.
    The standard representation consists of $n$ lists with up to $\lambda$ elements each,
    so $R = \Oh(n \lambda)$.   However, the lists can be shorter in general.
    Also, if the threshold probability $\frac1z$ is specified, at each position of a weighted sequence
    it suffices to store letters with probability at least $\frac1z$, and clearly
    there are at most $z$ such letters for each position. This reduction can be performed in linear time, so we shall always assume
    that $\lambda \le z$.
        

    The \emph{probability of matching} of a string $S$ with a weighted sequence $X$, $|S|=|X|=m$, is 
    $$\P(S,X) = \prod_{i=1}^m \pi^{(X)}_i(S[i]).$$
    We say that a string $S$ \emph{matches} a weighted sequence $X$
    with probability at least \fr, denoted by $S \match X$, if $\P(S,X) \ge \frac1z$.
    Given a weighted sequence $T$, by $T[i..j]$ we denote weighted sequence,
    called a \emph{factor} of $T$, equal to $T[i] \ldots T[j]$ (if $i>j$, then the factor is empty).
    We then say that a string $P$ \emph{occurs} in $T$ at position $i$ if $P$ matches the factor $T[i..i+m-1]$.
    We also say that $P$ is a \emph{\fr-solid factor} of $T$ at position $i$ (a \emph{\fr-solid prefix} if $i=1$
    and a \emph{\fr-solid suffix} if $j=n$).
    We denote the set of all positions where $P$ occurs in $T$ by $\Occ_\fr(P,T)$.

    \defdsproblem{\WPM}{
      A string $P$ of length $m$ and a weighted sequence $T$ of length $n$ with at most $\lambda$ letters at each position
      and $R$ in total, and a threshold probability \fr.
    }{
      The set $\Occ_\fr(P,T)$.
    }


  \section{Profile Matching and Weighted Pattern Matching}\label{sec:EWPM}
    In this section we present a solution to the \PM problem.
    Afterwards, we show that it can be applied for \WPM as well.

    For a scoring matrix $P$, the \emph{heavy string} of $P$, denoted $\H(P)$,
    is constructed by choosing at each position the heaviest letter, that is,
    the letter with the maximum score (breaking ties arbitrarily).
    Intuitively, $\H(P)$ is a string that matches $P$ with the maximum score.

    \begin{observation}\label{obs:crucial_profile}
      If we have $\Score(S,P) \ge Z$ for a string $S$ of length $m$ and an $m \times \sigma$ scoring matrix $P$,
      then $\HammingDistance(\H(P),S) \le \floor{\log M}$ where $M = \NumStrings_Z(P)$.
    \end{observation}
    \begin{proof}
      Let $d = \HammingDistance(\H(P),S)$.
      We can construct $2^{d}$ strings of length $|S|$ that match $P$ with a score above $Z$
      by taking either of the letters $S[j]$ or $\H(P)[j]$ at each position $j$ such that $S[j]\ne \H(P)[j]$.
      Hence, $2^{d} \le M$, which concludes the proof.
    \mayqed\end{proof}

    Our solution for the \PM problem works as follows.
    We first construct $P' = \H(P)$ and the data structure for finding lcp values between suffixes of $P'$ and $T$.
    Let the variable $s$ store the matching score of $P'$.
    In the $p$-th step, we calculate the matching score of $T[p..p+m-1]$ by iterating through
    subsequent mismatches between $P'$ and $T[p..p+m-1]$ and making adequate updates in the matching score $s'$, which
    starts at $s'=s$.
    The mismatches are found using lcp-queries.
    This process terminates when the score $s'$ drops below $Z$ or when all the mismatches have been found.
    In the end, we include $p$ in $\Occ_Z(P,T)$ if $s' \ge Z$.
    A pseudocode of this approach is given below for completeness.

    \begin{theorem}
      \PM problem can be solved in $\Oh(m\sigma + n \log M)$ time.
    \end{theorem}
    \begin{proof}
      Let us bound the time complexity of the presented algorithm.
      The heavy string $P'$ can be computed in $\Oh(m\sigma)$ time.
      The data structure for $\lcp$-queries in $P'T$ can be constructed in $\Oh(n+m)$ time by \cref{fct:ver}.
      Each query for $\lcp(P'[i..m],T[j..n])$ can then be answered in constant time by a corresponding $\lcp$-query
      in $P'T$, potentially truncated to the end of $P'$.
      Finally, for each position $p$ in the text $T$ we will consider at most $\floor{\log M}+1$ mismatches between $P'$ and $T$,
      as afterwards the score $s'$ drops below $Z$ due to \cref{obs:crucial_profile}.
    \mayqed\end{proof}

      \begin{procedure}[htpb]
      $m:=|P|$;\ \ $n:=|T|$;\ \ $\Occ:=\emptyset$\;
      $P' := \H(P)$\;
      Compute the data structure for $\lcp$-queries in $P'T$\;
      $s := \sum_{j=1}^m P[j,P'[j]]$\;
      \For{$p:=1$ \KwSty{to} $n-m+1$}{
        $s':=s$;\ \  $i:=1$;\ \ $j:=p$\;
        \While{$s' \ge Z$ \KwSty{and} $i \le m$}{
          $\Delta := \lcp(P'[i..m],T[j..n])$\;
          $i:=i+\Delta+1$;\ \ $j:=j+\Delta+1$\;
          \If{$i \le m+1$}{
            $s' := s' + P[i-1,T[j-1]] - P[i-1,P'[i-1]]$;
          }
        }
        \lIf{$s' \ge Z$}{insert $p$ to $\Occ$}
      }
      \Return{$\Occ$}\;
      \caption{ProfileMatching($P$, $T$, $Z$)}
    \end{procedure}
    
    Basically the same approach can be used for \WPM.
    In a natural way, we extend the notion of a heavy string to weighted sequences.
    Now we can restate \cref{obs:crucial_profile} in the language of probabilities instead of scores:
    \begin{observation}\label{obs:crucial}
      If a string $P$ matches a weighted sequence $X$ of the same length with probability at least \fr,
      then $\HammingDistance(\H(X),P) \le \floor{\log z}$.
    \end{observation}

    Comparing to the solution to \PM, we compute the heavy string of the text instead
    of the pattern.
    An auxiliary variable $\alpha$ stores the matching probability between a factor
    of $\H(T)$ and the corresponding factor of $T$; it needs to be updated when we move to the next position
    of the text.
    In the implementation, we perform the following operations on a weighted sequence:
    \begin{itemize}
      \item computing the probability of a given letter at a given position,
      \item finding the letter with the maximum probability at a given position.
    \end{itemize}

      \begin{procedure}[h]
      \caption{WeightedPatternMatching($P$, $T$, \fr)}
      $m:=|P|$;\ \ $n:=|T|$;\ \ $\Occ:=\emptyset$\;
      $T' := \H(T)$\;
      Compute the data structure for $\lcp$-queries in $PT'$\;
      $\alpha := \prod_{j=1}^m \pi^{(T)}_j(T'[j])$\;
      \For{$p:=1$ \KwSty{to} $n-m+1$}{
        $\alpha':=\alpha$;\ \ $i:=1$;\ \ $j:=p$\;
        \While{$\alpha' \ge \frac1z$ \KwSty{and} $i \le m$}{
          $\Delta := \lcp(P[i..m],T'[j..n])$\;
          $i:=i+\Delta+1$;\ \ $j:=j+\Delta+1$\;
          \If{$i \le m+1$}{
            $\alpha' := \alpha'\,\cdot\,\pi^{(T)}_{j-1}(P[i-1]) \,/\, \pi^{(T)}_{j-1}(T'[j-1])$;
          }
          }
        \lIf{$\alpha' \ge \frac1z$}{insert $p$ to $\Occ$}
        \If{$p \le n-m$}{
          $\alpha := \alpha \,\cdot\, \pi^{(T)}_{p+m}(T'[p+m]) \,/\, \pi^{(T)}_p(T'[p])$\;
        }
      }
      \Return{$\Occ$}\;
    \end{procedure}

    In the standard list representation, the latter can be performed on a single weighted sequence
    in $\Oh(1)$ time after $\Oh(R)$-time preprocessing.
    We can perform the former in constant time if, in addition to the list representation,
    we store the letter probabilities in a dictionary implemented using perfect hashing \cite{DBLP:journals/jacm/FredmanKS84}.
    This way, we can implement the algorithm in $\Oh(n \log z + R)$ time w.h.p.
    Alternatively, a deterministic dictionary~\cite{DBLP:conf/icalp/Ruzic08} can be used to obtain a deterministic solution in $\Oh(R\log^2\log \lambda + n\log z)$ time.
    We arrive at the following result.

    \begin{theorem}\label{thm:wpm}
      \WPM can be solved in $\Oh(R+n \log z)$ time with high probability by a Las-Vegas algorithm
      or in $\Oh(R\log^2 \log \lambda + n\log z)$ time deterministically.
    \end{theorem}

    \begin{remark}
      In the same complexity one can solve the \GWPM problem with a solid text.
    \end{remark}

 \section{Profile Consensus as Multichoice Knapsack}\label{sec:MK}
 Let us start with a precise statement of the \MK problem.
    \defdsproblempar{\MK}{
      A set $\C$ of $N$ items partitioned into $n$ disjoint classes $C_i$, each of size at most $\lambda$,
      two integers $v(c)$ and $w(c)$ for each item $c\in \C$, and two thresholds $V$ and $W$.
    }{
      Does there exist a \emph{choice} $S$ (a set $S\sub \C$ such that $|S\cap C_i|=1$ for each $i$)
      satisfying both
      $\sum_{c\in S} v(c) \le V$ and 
      $\sum_{c\in S} w(c) \le W$?
    }{
      $A_V$ and $A_W$: the number of choices $S$ satisfying $\sum_{c\in S} v(c) \le V$
      and $\sum_{c\in S} w(c) \le W$, respectively, as well as $A = \max(A_V,A_W)$ and $a=\min(A_V,A_W)$.
    }

  Indeed we see that the \PC problem reduces to the \MK problem.
  For two $m \times \sigma$ scoring matrices, we construct $n=m$ classes of $\lambda=\sigma$ items each,
  with values equal to the scores of the letters in the first matrix and weights equal to the scores in the second matrix;
  both thresholds $V$ and $W$ are equal to $Z$.
    
    For a fixed instance of \MK, we say that $S$ is a \emph{partial choice} if $|S\cap C_i|\le 1$ for each class.
    The set $D=\{i : |S\cap C_i|=1\}$ is called its \emph{domain}.
    For a partial choice $S$, we define $v(S) = \sum_{c \in S} v(c)$ and $w(S) = \sum_{c \in S} w(c)$.

     The classic $\Oh(2^{n/2})$-time solution to the \Knapsack problem~\cite{DBLP:journals/jacm/HorowitzS74}
 partitions $D=\{1,\ldots,n\}$ into two domains $D_i$ of size roughly $n/2$,
 and for each $D_i$ it generates all partial choices $S$ ordered by $v(S)$.
 Hence, it reduces the problem to an instance of \MK with $2$ classes.
 It is solved using the following lemma, proved below for completeness.
  \begin{lemma}\label{lem:knap2}
    The \MK problem can be solved in $\Oh(N)$ time if $n=2$ and
    the elements $c$ of $C_1$ and $C_2$ are sorted by $v(c)$.
  \end{lemma}
    \begin{proof}
      Since the items of $C_1$ and $C_2$ are sorted by $v(c)$, 
      a single scan through these items lets us remove all irrelevant elements, that is, elements dominated by other elements in their class.
      Next, for each $c_1\in C_1$ we compute $c_2\in C_2$ such that $v(c_2)\le V-v(c_1)$ but otherwise $v(c_2)$ is largest possible.
      As we have removed irrelevant elements from $C_2$, this item also minimizes $w(c_2)$ among all elements satisfying $v(c_2)\le V-v(c_1)$.
      Hence, if there is a feasible solution containing $c_1$, then $\{c_1,c_2\}$ is feasible. 
      If we process elements $c_1$ by non-decreasing values $v(c_1)$, the values $v(c_2)$ do not increase, and
      thus the items $c_2$ can be computed in $\Oh(N)$ time in total.
    \end{proof} 

  The same approach generalizes to \MK. 
  The partition is chosen to balance the number of partial choices
  in each domain, and the worst-case time complexity is $\Oh(\sqrt{Q\lambda})$, where $Q=\prod_{i=1}^n |C_i|$ is the number of choices.
  
  Our aim in this section is to replace $Q$ with the parameter $a$ (which never exceeds $Q$).
  The overall running time is going to be  $\Oh(N+\sqrt{a\lambda}\log A)$: an overhead  of $\Oh(\log A)$ appears.
  
  Two challenges arise once we adapt the meet-in-the-middle approach:
  how to restrict the set of partial choices to be generated so that a feasible solution is not missed,
  and how to define a partition $D=D_1\cup D_2$ to balance the number of partial choices generated for $D_1$ and $D_2$.
  A natural idea to deal with the first issue is to consider only partial choices with small values $v(S)$ or $w(S)$.
  This is close to our actual solution, which is based on the notion of \emph{ranks} of partial choices.
  Our approach to the second problem is to consider multiple partitions: those of the form $D=\{1,\ldots,j\}\cup\{j+1,\ldots,n\}$ for $1\le j \le n$.
  This results in an extra $\Oh(n)$ factor in the time complexity.
  However, in \cref{ss} we introduce preprocessing that lets us assume that $n=\Oh(\frac{\log A}{\log \lambda})$.
  While dealing with these two issues, some further effort is required to avoid few other extra terms in the running time.
  In case of our algorithm, this is only $\Oh(\log \lambda)$, which stems from the fact that we need to keep partial solutions ordered by $v(S)$.

    For a partial choice $S$, we define $\rank_V(S)$ as the number of partial choices $S'$ with the same domain for which $v(S')\le v(S)$. 
    We symmetrically define $\rank_W(S)$.
    Ranks are introduced as an analogue of probabilities in weighted sequences.
    Probabilities are multiplicative, while for ranks we have submultiplicativity:
      
    \begin{fact}\label{fct:comb}
      Assume that $S=S_1\cup S_2$ is a decomposition of a partial choice $S$ into two disjoint subsets.
      Then $\rank_V(S_1)\rank_V(S_2)\le \rank_V(S)$ (and same for $\rank_W$).
    \end{fact}
    \begin{proof}
      Let $D_1$ and $D_2$ be the domains of $S_1$ and $S_2$, respectively.
      For every partial choices $S'_1$ over $D_1$ and $S'_2$ over $D_2$ such that $v(S'_1) \le v(S_1)$ and $v(S'_2) \le v(S_2)$, we have
      $v(S'_1 \cup S'_2)=v(S'_1)+v(S'_2)\le v(S)$.
      Hence, $S'_1\cup S'_2$ must be counted while determining $\rank_V(S)$.
    \end{proof}

   For $0\le j \le n$, let $\L_j$ be the list of partial choices with domain $\{1,\ldots,j\}$ ordered by value $v(S)$,
   and for $\ell>0$ let $V^{(\ell)}_{\L_j}$ be the value $v(S)$ of $\ell$-th element of $\L_j$ ($\infty$ if $\ell>|\L_j|$).
  Analogously, for $1\le j \le n+1$, we define $\R_j$ as the list of partial choices over $\{j,\ldots,n\}$ ordered by $v(S)$,
  and for $r>0$, $V^{(r)}_{\R_j}$ as the value of the $r$-th element of $\R_j$ ($\infty$ if $r>|\R_j|$).
  
  The following two observations yield a decomposition of each choice into a single item and two partial solutions of a small rank.
  In particular, we do not need to know $A_V$ in order to check if the ranks are sufficiently large.
  
  \begin{lemma}\label{lem:decomp}
  Let $\ell$ and $r$ be positive integers such that $V^{(\ell)}_{\L_j}+V^{(r)}_{\R_{j+1}}> V$ for each $0\le j \le n$.
  For every choice $S$ with $v(S)\le V$, there is an index $j\in\{1,\ldots,n\}$ and a decomposition $S=L\cup\{c\}\cup R$
  such that $v(L) < V^{(\ell)}_{\L_{j-1}}$, $c\in C_j$, and $v(R) < V^{(r)}_{\R_{j+1}}$.
  \end{lemma}
  \begin{proof}
    Let $S=\{c_1,\ldots,c_n\}$ with $c_i\in C_i$ and, for $0\le i \le n$, let $S_i = \{c_1,\ldots,c_i\}$.
    If $v(S_{n-1})< V^{(\ell)}_{\L_{n-1}}$, we set $L=S_{n-1}$, $c=c_n$, and $R=\emptyset$, satisfying the claimed conditions.
    
    Otherwise, we define $j$ as the smallest index $i$ such that $v(S_i) \ge V^{(\ell)}_{\L_i}$,
    and we set $L=S_{j-1}$, $c=c_j$, and $R=S\setminus S_j$.
    The definition of $j$ implies $v(L)<V^{(\ell)}_{\L_{j-1}}$ and $v(L\cup\{c\})\ge V^{(\ell)}_{\L_j}$.
    Moreover, we have $v(L\cup \{c\})+v(R)=v(S)\le V < V^{(\ell)}_{\L_j}+V^{(r)}_{\R_{j+1}},$
    and thus $v(R) < V^{(r)}_{\R_{j+1}}$. 
  \end{proof}
  
  \begin{fact}\label{fct:bound}
  Let $\ell,r>0$.
  If $V^{(\ell)}_{\L_j}+V^{(r)}_{\R_{j+1}}\le  V$ for some $j\in\{0,\ldots,n\}$, then $\ell \cdot r \le A_V$.
  \end{fact}
  \begin{proof}
  Let $L$ and $R$ be the $\ell$-th and $r$-th entry in $\L_j$ and $\R_{j+1}$, respectively.
  Note that $v(L\cup R)\le  V$ implies
  $\rank_V(L\cup R) \le A_V$ by definition of~$A_V$.
  Moreover, $\rank_V(L)\ge \ell$ and $\rank_V(R)\ge r$
  (the equalities may be sharp due to draws).
  Now, \cref{fct:comb} yields the claimed bound.
  \end{proof}

   Note that $\L_j$ can be obtained by interleaving $|C_j|$ copies of $\L_{j-1}$, where each copy corresponds
  to extending the choices from $\L_{j-1}$ with a different item.
  If we were to construct $\L_{j}$ having access to the whole $\L_{j-1}$, we could proceed as follows.
  For each $c\in C_j$, we maintain an \emph{iterator} on $\L_{j-1}$
  pointing to the first element $S$ on $\L_{j-1}$ for which $S\cup\{c\}$ has not yet been added to $\L_{j}$.
  The associated \emph{value} is $v(S\cup\{c\})$.
  All iterators initially point at the first element of $\L_{j-1}$. 
  Then the next element to append to $\L_j$ is always $S\cup\{c\}$ corresponding to the iterator with minimum value. 
  Having processed this partial choice, we advance the pointer (or remove it, once it has already scanned the whole $\L_{j-1}$).
  This process can be implemented using a binary heap $H_j$ as a priority queue,
  so that initialization requires $\Oh(|C_j|)$ time and outputting a single element takes $\Oh(\log |C_j|)$ time.

  For all $r \ge 0$, let $\L^{(r)}_j$ be the prefix of $\L_j$ of length $\min(r,|\L_j|)$ and $\R^{(r)}_j$ be the prefix of $\R_j$ of length $\min(r, |\R_j|)$.
  A technical transformation of the procedure stated above
  leads to an online algorithm that constructs the prefixes $\L^{(r)}_j$ and $\R^{(r)}_j$.
  Along with each reported partial choice $S$, the algorithm also computes $w(S)$.

  \begin{lemma}\label{lem:generate}
    After $\Oh(N)$-time initialization, one can construct $\L^{(i)}_1,\ldots,\L^{(i)}_n$
    online for $i=0,1,\ldots$, spending $\Oh(n\log \lambda)$ time per each step.
    Symmetrically, one can construct $\R^{(i)}_1,\ldots,\R^{(i)}_n$ in the same time complexity.
  \end{lemma}
   \begin{proof}
      Our online algorithm is going to use the same approach as the offline computation of lists $\L_j^{(r)}$.
      The order of computations is going to be changed, though.
      
      At each step, for $j=1$ to $n$ we shall extend lists $\L^{(i-1)}_j$ with a single element 
      (unless the whole $\L_j$ has been already generated) from the top of the heap $H_j$.
      Note that this way each iterator in $H_j$ always points to an element 
      that is already in $\L^{(i-1)}_{j-1}$ or to the first element that has not been yet added to $\L_{j-1}$,
      which is represented by the top of the heap $H_{j-1}$. 
      
      We initialize the heaps as follows: we introduce $H_0$ which represents the empty choice $\emptyset$ with $v(\emptyset)=0$.
      Next, for $j=1,\ldots,n$ we build the heap $H_j$ representing $|C_j|$ iterators initially pointing to the top of $H_{j-1}$.
      The initialization takes $\Oh(N)$ time in total since a binary heap can be constructed in time linear in its size.
      
      At each step, the lists $\L^{(i-1)}_j$ are extended for consecutive values $j$ from $1$ to $n$. 
      Since $\L^{(i-1)}_{j-1}$ is extended before $\L^{(i-1)}_j$, all iterators in $H_j$ point to the elements of $\L^{(i)}_{j-1}$ while we compute $\L_j^{(i)}$.
      We take the top of $H_j$ and move it to $\L^{(i)}_j$.
      Next, we advance the corresponding iterator and update its position in the heap $H_j$.
      After this operation, the iterator might point to the top of $H_{j-1}$.
      If $H_{j-1}$ is empty, this means that the whole list $\L_{j-1}$ has already been generated and traversed by the iterator.
      In this case, we remove the iterator.
      
      It is not hard to see that this way we indeed simulate the previous offline solution.
      A single phase makes $\Oh(1)$ operations on each heap $H_j$.
      The running time is bounded by $\Oh(\sum_{j} \log |C_j|)=\Oh(n\log \lambda)$ at each step of the algorithm.
    \end{proof}

  The reduction of the following lemma is presented in \cref{ss}.
  Note that we may always assume that $\lambda \le a \le A$.
  Indeed, if we order the items $c\in C_i$ according to $v(c)$, 
  then only the first $A_V$ of them might belong to a choice $S$ with $v(S)\le V$. 

     \begin{lemma}\label{lem:knapred2}
    Given an instance $I$ of the \MK problem, one can compute in $\Oh(N+\lambda\log A)$ time
    an equivalent instance $I'$ with $A_V'\le A_V$, $A_W'\le A_W$,
    $\lambda'\le \lambda$, and $n'=\Oh(\frac{\log A}{\log \lambda})$.
    \end{lemma}

  \begin{theorem}\label{thm:knap}
    \MK can be solved in $\Oh(N+\sqrt{a\lambda}\log A)$ time.
  \end{theorem}
  \begin{proof}
  Below, we give an algorithm in $\Oh(N+\sqrt{A_V\lambda}\log A)$ time. 
  The final solution runs it in parallel on the original instance and on the instance with $v$ and $V$ swapped with $w$ and~$W$,
   waiting until at least one of them terminates.
    
    We increment an integer $r$ starting from~$1$, maintaining $\ell=\ceil{\frac{r}{\lambda}}$ and the lists $\L_{j}^{(\ell)}$ and $\R_{j+1}^{(r)}$ for $0\le j \le n$, as long as $V^{(\ell)}_{\L_j}+V^{(r)}_{\R_{j+1}}\le V$ for some $j$ (or until all the lists have been completely generated).
    By \cref{fct:bound}, we stop at $r=\Oh(\sqrt{A_V  \lambda})$.
    \cref{lem:knapred2} lets us assume that $n=\Oh(\frac{\log A}{\log \lambda})$, so  
    the running time of this phase is $\Oh(N+\sqrt{A_V  \lambda}\log A)$ due to \cref{lem:generate}.

    Due to \cref{lem:decomp}, every feasible solution $S$ admits a decomposition $S=L\cup\{c\}\cup R$
    with $L\in \L_{j-1}^{(\ell)}$, $c\in C_j$, and $R\in \R_{j+1}^{(r)}$ for some index $j$. 
    We consider all possibilities for $j$. For each of them we will reduce searching for $S$ to an instance of the \MK
    problem with $N'=\Oh(\sqrt{A_V\lambda})$ and $n'=2$. By \cref{lem:knap2}, these instances can be solved in $\Oh(n\sqrt{A_V\lambda})=\Oh(\sqrt{A_V\lambda}\frac{\log A}{\log \lambda})$  time in total.
   
    The items of the $j$-th instance are going to belong to classes $\L_{j-1}^{(\ell)}\odot C_j$ and $\R_{j+1}^{(r)}$,
    where $\L_{j-1}^{(\ell)}\odot C_j = \{L\cup \{c\} : L\in \L_{j-1}^{(\ell)} , c\in C_j\}$.
    The set $\L_{j-1}^{(\ell)}\odot C_j$ can be constructed by merging $|C_j|\le \lambda$ sorted lists each of size $\ell=\Oh(\sqrt{A_V/\lambda})$,
    i.e., in $\Oh(\sqrt{A_V\lambda}\log \lambda)$ time.
    Summing up over all indices $j$, this gives $\Oh(\sqrt{A_V \lambda}\log\lambda \frac{\log A}{\log \lambda})=\Oh(\sqrt{A_V \lambda}\log A)$ time.
       
    Clearly, each feasible solution of the constructed instances represents a feasible solution of the initial instance,
    and by \cref{lem:decomp}, every feasible solution of the initial instance has its counterpart in one of the constructed instances. 
  \end{proof}
 
  \subsection{Proof of Lemma~\ref{lem:knapred2}}\label{ss}
  Our reduction consists of two steps.   Its implementation uses the following notions:
  For each class $C_i$, let $v_{\min}(i) = \min\{v(c) : c\in C_i\}$. 
  Also, let $V_{\min} = \sum_{i=1}^n v_{\min}(i)$; note that $V_{\min}$ is the smallest possible value $v(S)$ of a choice $S$.
  We symmetrically define $w_{\min}(i)$ and $W_{\min}$.
  First, we make sure that $n=\Oh(\log A)$.
  \begin{lemma}\label{lem:knapred}
    Given an instance $I$ of the \MK problem, one can compute in linear time an equivalent instance $I'$ with $N'\le N$, $A_V'\le A_V$, $A_W'\le A_W$, $\lambda'\le \lambda$, and $n' \le 2\log A$.
  \end{lemma}
  \begin{proof}   
    Observe that if some class $C_i$ contains a single item $c$ for which both $v(c)=v_{\min}(i)$
    and $w(c)=w_{\min}(i)$, then we can greedily include it in the solution $S$.  
    Hence, we can remove such a class, setting $V := V- v_{\min}(i)$ and $W := W- w_{\min}(i)$.
    We execute this reduction rule exhaustively, which clearly takes $\Oh(N)$ time in total
    and may only decrease the parameters $A_V$ and $A_W$.
    After the reduction, each class contains at least two items.

    We shall prove that now we can either find out that $A \ge 2^{n/2}$ or 
    that we are dealing with a NO-instance.
    To decide which case holds, let us define $\Delta_V(i)$ as the difference between the second smallest
    value in the multiset $\{v(c) : c\in C_i\}$ and $v_{\min}(i)$.
    We set $\Delta_V^{\mmid}$ as the sum of the $\ceil{\frac{n}{2}}$ smallest values $\Delta_V(i)$
    for $1\le i \le n$; analogously we define $\Delta_W^{\mmid}$.
   
    \begin{claim}
      If $V_{\min} + \Delta_V^{\mmid} \le V$, then $A_V \ge 2^{n/2}$;
      if $W_{\min} + \Delta_W^{\mmid} \le W$, then $A_W \ge 2^{n/2}$;
      otherwise, we are dealing with a NO-instance.
    \end{claim}
    \begin{proof}
      First, assume that $V_{\min} + \Delta_V^{\mmid} \le V$.
      This means that there is a choice $S$ with $v(S)\le V$ containing at least $\frac{n}{2}$ items $c$ such that $\rank_V(c)\ge 2$.
      \cref{fct:comb} yields $\rank_V(S)\ge 2^{\ceil{n/2}}$ and consequently $A_V \ge 2^{n/2}$, as claimed.
      Symmetrically, if $W_{\min} + \Delta_W^{\mmid} \le W$, then $A_W \ge 2^{n/2}$.
   
      Now, suppose that there is a feasible solution $S$.
      As no class contains a single item minimizing both $v(c)$ and $w(c)$, there are at least $\ceil{\frac{n}{2}}$
      classes for which $S$ contains an item not minimizing $v(c)$, or at least $\ceil{\frac{n}{2}}$ classes
      for which $S$ contains an item not minimizing $w(c)$. Without loss of generality, we assume that the former holds.
      Let $D$ be the set of at least $\ceil{\frac{n}{2}}$ classes $i$ satisfying the condition. 
      If $c\in C_i$ does not minimize $v(c)$, then $v(c)\ge v_{\min}(i)+\Delta_V(i)$. 
      Consequently, $V\ge v(S) = V_{\min} + \sum_{i\in D} \Delta_V(i)$. 
      However, observe that $ \sum_{i\in D} \Delta_V(i) \ge \Delta_V^{\mmid}$,
      so $V \ge V_{\min} + \Delta_V^{\mmid}$, as claimed.
    \end{proof}
   
    The conditions from the claim can be verified in $\Oh(N)$ time using a linear-time selection algorithm to compute $\Delta_V^{\mmid}$ and $\Delta_W^{\mmid}$.
    If any of the first two conditions holds, we return the instance obtained using our reduction.
    Otherwise, we output a dummy NO-instance.
  \end{proof}

    Before we proceed with the second reduction,   let us introduce an auxiliary notion.
    An item $c\in C_j$ is \emph{irrelevant} if there is another item $c'\in C_j$ that \emph{dominates} $c$, i.e., 
    such that $v(c)>v(c')$ and $w(c)>w(c')$.
    Removing irrelevant items leads to an equivalent instance of the \MK problem, and it may only decrease the parameters.
    
   \begin{lemma}\label{lem:redstep}
   Consider a class of items in an instance of the \MK problem.
   In linear time, we can remove some irrelevant items from the class so that the resulting class $C$
   satisfies $\max(\rank_V(c),\rank_W(c)) > \frac13 |C|$ for each item $c\in C$.
   \end{lemma}
   \begin{proof}
   First, note that using a linear-time selection algorithm, we can determine for each item $c$
   whether $\rank_V(c)\le \frac13|C|$ and whether $\rank_W(c)\le \frac13|C|$. 
   If there is no item satisfying both conditions, we keep $C$ unaltered.
   Otherwise, we have an item which dominates at least $|C|-\rank_V(c)-\rank_W(c) \ge \frac13|C|$
   other items. We scan through all items in $C$ and remove those dominated by $c$.
   Next, we repeat the algorithm.
   The running time of a single phase is clearly linear, and since $|C|$ decreases geometrically,
   the total running time is also linear.
   \end{proof}    
   
   A straightforward way to decrease the number of classes is to replace two distinct classes $C_i$, $C_j$
   with their Cartesian product $C_i \times C_j$, 
   assuming that the weight of a pair $(c_i,c_j)$ is the sum of weights of $c_i$ and $c_j$. 
   This clearly leads to an equivalent instance of the \MK problem, does not alter the parameters $A_V$, $A_W$, and decreases $n$. 
   On the other hand $N$ and $\lambda$ may increase; the latter happens only if $|C_i| \cdot |C_j| > \lambda$.
      
   These two reduction rules let us implement our reduction procedure which constitutes the proof of Lemma~\ref{lem:knapred2}.
    \begin{proof}
    First, we apply \cref{lem:knapred} to make sure that $n\le 2\log A$ and $N = \Oh(\lambda \log A)$.
    We may now assume that $\lambda \ge 3^6$, as otherwise we already have $n = \Oh(\frac{\log A}{\log \lambda})$.
     
    Throughout the algorithm, whenever there are distinct classes of size at most $\sqrt{\lambda}$, 
    we shall replace them with their Cartesian product.
    This may happen only $n-1$ times, and a single execution takes $\Oh(\lambda)$ time,
    so the total running time needed for this part is $\Oh(\lambda \log A)$.
    
    Furthermore, for every class that we get in the input instance or obtain as a Cartesian product,
    we apply \cref{lem:redstep}.  The total running time spent on this is also $\Oh(\lambda \log A)$.
    
    Having exhaustively applied these reduction rules, we are guaranteed that
    $\max(\rank_V(c),\rank_W(c))>\frac13\sqrt{\lambda}\ge \lambda^{\frac13}$ for items $c$ from all but one class. 
    Without loss of generality, we assume that the classes satisfying this condition are $C_1,\ldots,C_k$.
    
    Recall that $v_{\min}(i)$ and $w_{\min}(i)$ are defined as minimum values and weights of items in class $C_i$
    and that $V_{\min}$ and $W_{\min}$ are their sums over all classes.
    For $1\le i \le k$, we define $\Delta_V(i)$ as the difference between the 
    $\big\lceil{\lambda^{\frac13}}\big\rceil$-th smallest value in the multiset  $\{v(c) : c\in C_i\}$ and $v_{\min}(i)$.
    Next, we define $\Delta_V^{\mmid}$ as the sum of the $\ceil{\frac{k}{2}}$ smallest values $\Delta_V(i)$.
    Symmetrically, we define $\Delta_W(i)$ and $\Delta_W^{\mmid}$.
    We shall prove a claim analogous to that in the proof of \cref{lem:knapred}.
    
    \begin{claim}
    If $V_{\min} + \Delta_V^{\mmid}\le V$, then $A_V \ge \lambda^{\frac16 k}$; if $W_{\min} + \Delta_W^{\mmid}\le W$, then $A_W \ge \lambda^{\frac16 k}$;  otherwise, we are dealing with a NO-instance. 
    \end{claim}
    \begin{proof}
    First, suppose that $V_{\min} + \Delta_V^{\mmid}\le V$.
    This means that there is a choice $S$ with $v(S)\le V$ 
    which contains at least $\frac{k}{2}$ items $c$ with $\rank_V(c)\ge \lambda^{\frac13}$. 
    By \cref{fct:comb}, the rank of this choice is at least $\lambda^{\frac16 k}$, so $A_V \ge \lambda^{\frac16 k}$, as claimed.
    The proof of the second case is analogous.
    
    Now, suppose that there is a feasible solution $S=\{c_1,\ldots,c_n\}$. For $1\le i \le k$, we have $\rank_V(c_i)\ge \lambda^{\frac13}$
    or $\rank_W(c_i) \ge \lambda^{\frac13}$. Consequently, $\rank_V(c_i)\ge \lambda^{\frac13}$ holds for at least $\ceil{\frac{k}{2}}$
    classes or $\rank_W(c_i)\ge \lambda^{\frac13}$ holds for at least $\ceil{\frac{k}{2}}$ classes.
    Without loss of generality, we assume that the former holds. Let $D$ be the set of (at least $\ceil{\frac{k}{2}}$) classes $i$ satisfying the condition.
    For each $i\in D$, we clearly have $v(c_i)\ge v_{\min}(i)+\Delta_V(i)$, while for each $i\notin D$, we have $v(c_i)\ge v_{\min}(i)$.
    Consequently, $V\ge v(S) \ge V_{\min} + \sum_{i\in D} \Delta_V(i) \ge V_{\min} + \Delta_V^{\mmid}$.
    Hence, $V \ge V_{\min} + \Delta_V^{\mmid}$, which concludes the proof.
    \end{proof}
       
    The condition from the claim can be verified using a linear-time selection algorithm: 
    first, we apply it for each class to compute $\Delta_V(i)$ and $\Delta_W(i)$, 
    and then, globally, to determine $\Delta_V^{\mmid}$ and $\Delta_W^{\mmid}$.
    If one of the first two conditions hold, we return the instance obtained through the reduction.
    It satisfies $A \ge \lambda^{\frac 16 k}$, i.e., $n \le 1+k \le 1+6\frac{\log A}{\log \lambda}$.
    Otherwise, we construct a dummy NO-instance.
    \end{proof}

  \section{Weighted Consensus and General Weighted Pattern Matching}\label{sec:GWPMReduction}
    The \WC problem is formally defined as follows.
    \defdsproblem{\WC}{
      Two weighted sequences $X$ and $Y$ of length $n$ with at most $\lambda$ letters at each position and $R$ in total,
      and a threshold probability \fr.
    }{
      A string $S$ such that $S \match X$ and $S \match Y$ or NONE if no such string exists.
    }

    If two weighted sequences satisfy the consensus, we write $X \match Y$ and say that $X$ \emph{matches} $Y$
    with probability \fr.
    With this definition of a match, we extend the notion of an occurrence and the notation $\Occ_\fr(P,T)$ to arbitrary weighted sequences.

    \defdsproblem{\GWPMFull (\GWPM)}{
      Two weighted sequences $P$ and $T$ of length $m$ and $n$, respectively, with at most $\lambda$ letters at each position
      and $R$ in total, and a threshold probability \fr.
    }{
      The set $\Occ_\fr(P,T)$.
    }

    In the case of the \GWPM problem, it is more useful to provide an \emph{oracle} that finds solid factors that correspond to particular occurrences of the pattern.
    Such an oracle, given $i \in \Occ_\fr(P,T)$, computes a string that matches both $P$ and $T[i..i+m-1]$.


    We say that a string $P$ is a \emph{maximal \fr-solid prefix} of a weighted sequence $X$
    if $P$ is a \fr-solid prefix of $X$ and no string $P' = Ps$, for $s \in \Sigma$, is a \fr-solid prefix of $X$.
    Our algorithms rely on the following simple combinatorial observation, originally due to Amir et al.\ \cite{amir_weighted_property_matching_j}.

    \begin{fact}[\cite{amir_weighted_property_matching_j}]\label{fct:maxprefixes}
      A weighted sequence has at most $z$ different maximal \fr-solid prefixes.
    \end{fact}

    The \WC problem is actually a special case of \MK.
    Namely, given an instance of the former, we can create an instance of the latter with $n$ classes $C_i$,
    each containing an item $c_{i,s}$ for every letter $s$ which has non-zero probability at position $i$ in both $X$ and $Y$.
    We set $v(c_{i,s})=-\log \pi^{(X)}_i(s)$ and $w(c_{i,s})=-\log \pi^{(Y)}_i(s)$ for this item,
    whereas the thresholds are $V=W=\log z$. It is easy to see that this reduction indeed yields an equivalent instance
    and that it can be implemented in linear time.
    By \cref{fct:maxprefixes}, we have $A\le z$ for this instance, so \cref{thm:knap} yields the following result:
    \begin{corollary}\label{cor:red_simple}
    \WC problem can be solved in $\Oh(R+\sqrt{z\lambda}\log z)$ time.
    \end{corollary}
    
    The \GWPM problem can be clearly reduced to $n+m-1$ instances of \WC. 
    This leads to a naive $\Oh(nR + n\sqrt{z\lambda}\log z)$-time algorithm.
    Below, we remove the first term in this complexity. 
    Our solution applies the approach used in \cref{sec:EWPM} for \WPM and uses an observation
    analogous to \cref{obs:crucial}.
    
    \begin{observation}\label{obs:crucial2}
      If $X$ and $Y$ are weighted sequences that match with threshold $\fr$,
      then $\HammingDistance(\H(X),\H(Y)) \le 2\floor{\log z}$.
      Moreover there exists a consensus string $S$ such that $S[i] = \H(X)[i] = \H(Y)[i]$ unless $\H(X)[i]\ne \H(Y)[i]$.
    \end{observation}
 
    Our algorithm starts by computing $P'=\H(P)$ and $T'=\H(T)$ and the data structure for $\lcp$-queries in $P'T'$.
    We try to match $P$ with every factor $T[p..p+m-1]$ of the text.
    Following \cref{obs:crucial2}, we check if $\HammingDistance(T'[p..p+m-1],P') \le 2\floor{\log z}.$
    If not, then we know that no match is possible.
    Otherwise, let $D$ be the set of positions of mismatches between $T'[p..p+m-1]$ and $P'$.
    Assume that we store
    $\alpha = \prod_{j=1}^{m} \pi^{(T)}_{p+j-1}(T'[p+j-1])$ and $\beta = \prod_{j=1}^m \pi^{(P)}_j(P'[j]).$
    Then, in $\Oh(|D|)$ time we can compute $\alpha'=\prod_{j\notin D} \pi^{(T)}_{p+j-1}(T'[p+j-1])$ and
    $\beta' = \prod_{j \notin D} \pi^{(P)}_j(P'[j])$.
    Now, we only need to check what happens at the positions in $D$.

    If $D = \emptyset$, then it suffices to check if $\alpha \ge \frac1z$ and $\beta \ge \frac1z$.
    Otherwise, we construct two weighted sequences $X$ and $Y$ by selecting only the positions from $D$ in $T[p..p+m-1]$ and in $P$.
    We multiply the probabilities of all letters at the first position in $X$ by $\alpha'$ and in $Y$ by $\beta'$.
    It is clear that $X\match Y$ if and only if $T[p..p+m-1]\match P$.
    
    Thus, we reduced the \GWPM problem to at most $n-m+1$ instances of the \WC problem for strings of length $\Oh(\log z)$. 
    By \cref{cor:red_simple}, solving such an instance takes $\Oh(\lambda\log z + \sqrt{z\lambda}\log z)=\Oh(\sqrt{z\lambda}\log z)$ time.   
    Our reduction requires $\Oh(R\log^2 \log \lambda)$ time to preprocess the input (as in \cref{thm:wpm}),
    but this is dominated by the $\Oh(n\sqrt{z\lambda}\log z )$ total time of solving the \WC instances.
    If we memorize the solutions to all those instances
    together with the sets of mismatches $D$ that lead to those instances, 
    then we can also implement the oracle for the \GWPM problem with $\Oh(m)$-time queries.
     In \cref{app:SDWC}, we design a tailor-made solution to replace the generic algorithm for the \MK problem,
     which lets us improve the $\log z$ factor to $\log\log z + \log \lambda$. 
        
    The following reduction from \MK to \WC immediately yields that any significant improvement in the dependence
    on $z$ and $\lambda$ in the running time of our algorithm
    would lead to breaking long-standing barriers for special cases of \MK.
   \begin{lemma}\label{lem:red}
    Given an instance $I$ of the \MK problem with $n$ classes of size $\lambda$, in linear time one can construct
    an equivalent instance of the \WC problem with $z=\Oh(\prod_{i=1}^{n}|C_i|)$ and sequences of length $\Oh(n)$ over alphabet of size $\lambda$.
  \end{lemma}
    \begin{proof}
    We construct a pair of weighted sequences $X,Y$ of length $n$
    over alphabet $\Sigma=\{1,\ldots,\lambda\}$. 
    Intuitively, choosing letter $j$ at position $i$ will correspond to taking the $j$-th element of $C_i$ to the solution $S$,
    which we denote as $c_{i,j}$.
    Without loss of generality, we assume that weights and values are non-negative. Otherwise, we may subtract $v_{\min}(i)$ from $v(c_{i,j})$
    and $w_{\min}(i)$ from $w(c_{i,j})$ for each item $c_{i,j}$, as well $V_{\min}$ from $V$ and $W_{\min}$ from $W$.
    
    We set $M$ to the smallest power of two such that $M\ge\max(n, V, W)$.
    Let $p_i^{(X)}(j) = \log \pi_i^{(X)}(j)$ and $p_i^{(Y)}(j) = \log \pi_i^{(Y)}(j)$ for $j \in \Sigma$.
    For $j\in \{1,\ldots,|C_i|\}$, we set:
    $$p_i^{(X)}(j) = -\frac{\ceil{M\log|C_i|} + v(c_{i,j})}{M}, \quad p_i^{(Y)}(j)=-\frac{\ceil{M\log|C_i|} +w(c_{i,j})}{M}.$$
    Clearly, $\sum_{j=1}^{|C_i|} 2^{p_i^{(X)}(j)}\le 1$ and $\sum_{j=1}^{|C_i|} 2^{p_i^{(Y)}(j)}\le 1$.
    Moreover,
    we set 
    $$\log z_X = \frac1M \left(V + \sum_{i=1}^n\ceil{M\log|C_i|}\right)   \quad \text{and} \quad \log z_Y = \frac1M\left(W + \sum_{i=1}^n\ceil{M\log|C_i|}\right).$$
    By the choice of $M$, we have $\max(z_X,z_Y) \le 2^{\frac1M(\max(V,W)+n)}\prod_{i=1}^{n}|C_i|\le 4\prod_{i=1}^{n}|C_i|$. 
    
    This way, for a string $P$ of length $n$, we have 
    $$\log \P(P,X)=-\frac1M\left(\sum_{i=1}^n\ceil{M\log|C_i|}+\sum_{i=1}^n v(c_{i,P[i]})\right) \ge -\log z_X \; \Longleftrightarrow \; \sum_{i=1}^n
     v(c_{i,P[i]}) \le V$$
    and 
    $$\log \P(P,Y)=-\frac1M\left(\sum_{i=1}^n\ceil{M\log|C_i|}+\sum_{i=1}^n w(c_{i,P[i]})\right) \ge -\log z_Y \; \Longleftrightarrow \; \sum_{i=1}^n
     w(c_{i,P[i]}) \le W.$$
    
    Thus, $P$ is a solution to the constructed instance of the \WC problem with two threshold probabilities, $\frac{1}{z_X}$ and $\frac{1}{z_Y}$,
    if and only if $S = \{c_{i,j}\,:\,P[i]=j\}$ is a solution to the underlying instance of the \MK problem.
    To have a single threshold $z=\max(z_X,z_Y)$, we append an additional position $n+1$ with symbol 1 only, 
    with $p_{n+1}^{(X)}(1)=0$ and $p_{n+1}^{(Y)}(1)=\log z_Y - \log z_X$ provided that $z_X \ge z_Y$,
    and symmetrically otherwise.
    
    If one wants to make sure that the probabilities at each position sum up to exactly one, two further letters can be introduced,
    one of which gathers the remaining probability in $X$ and has probability 0 in $Y$, and the other gathers
    the remaining probability in $Y$, and has probability 0 in $X$. 
   \end{proof} 

  \begin{theorem}\label{thm:lb}
  \WC problem is NP-hard and cannot be solved in:
  \begin{enumerate}
   \item $\Ohstar(z^{o(1)})$ time unless the Exponential Time Hypothesis (ETH) fails;
   \item $\Ohstar(z^{0.5-\varepsilon})$ time for some $\varepsilon>0$, unless there is an $\Ohstar(2^{(0.5-\varepsilon)n})$-time algorithm for the \SubsetSum problem;
    \item $\Ohtilde(R+z^{0.5}\lambda^{0.5-\varepsilon})$ time for some $\varepsilon>0$ and for $n=\Oh(1)$, unless
    there is an $\Oh(\lambda^{2(1-\varepsilon)})$-time algorithm for 3-\Sum.
    \end{enumerate}
  \end{theorem}
    \begin{proof}
    We use \cref{lem:red} to derive algorithms for the \MK problem based on hypothetical solutions for \WC.
    \SubsetSum is a special case of \MK with $\lambda=2$, i.e., $\prod_{i}|C_i|=2^n$. Hence,
    an $\Ohstar(z^{o(1)})$-time solution for \WC would yield an $\Ohstar(2^{o(n)})$-time algorithm for \SubsetSum,
    which contradicts ETH by the results of Etscheid et al.~\cite{DBLP:conf/mfcs/EtscheidKMR15} and Gurari~\cite{DBLP:books/daglib/0069796}.
    Similarly, an $\Ohstar(z^{0.5-\varepsilon})$-time solution for \WC would yield an $\Ohstar(2^{(0.5-\varepsilon)n})$-time algorithm for \SubsetSum. 
    Moreover, $k$-\Sum is a special case of \MK with $n=k=\Oh(1)$, i.e., $\prod_{i}|C_i|=\lambda^{k}$.
    Hence, an $\Ohtilde(R+z^{0.5}\lambda^{0.5-\varepsilon})$-time solution for \WC yields
    an $\Oh(\lambda + \lambda^{1.5+0.5-\varepsilon})=\Oh(\lambda^{2-\varepsilon})$-time algorithm for 3-\Sum.
    \end{proof}

    Nevertheless, it might still be possible to improve the dependence on $n$ in the \GWPM problem.
    For example, one may hope to achieve $\Ohtilde(nz^{0.5-\varepsilon}+z^{0.5})$ time for $\lambda=\Oh(1)$.

  \section{Faster \GWPM via Short Dissimilar Weighted Consensus}\label{app:SDWC}
  This section provides a faster solution for the \GWPMFull problem.
  The key ingredient is an improved solution for the following \SDWCFull problem:
     \defdsproblem{\SDWCFull (\SDWC)}{
      A  threshold probability \fr\ and two weighted sequences $X$ and $Y$ of length $n\le 2\floor{\log z}$ with at most $\lambda\le z$ letters at each position
      and such that  $\H(X)$ and $\H(Y)$ are \emph{dissimilar}, i.e., $\H(X)[i] \ne \H(Y)[i]$ for each position~$i$.
    }{
      A string $S$ such that $S \match X$ and $S \match Y$ or NONE if no such string exists.
    }
    
    Note that the instances of the \WC problem produced by the reduction of \cref{sec:GWPMReduction} are actually instances of the \SDWC problem.
    Our tailor-made solution for the \SDWC problem works in $\Oh(\sqrt{z\lambda} (\log\log z + \log \lambda))$ time.    
    It assumes that the letters at each position of the weighted sequences are sorted according to probabilities
    (in addition to storing the dictionary of letters and probabilities).
    This can be achieved in $\Oh(\lambda \log \lambda)$ time for each position.
    We have just proved:

    \begin{lemma}\label{lem:sdwc}
      The \GWPM problem and the computation of its oracle can be reduced
      in $\Oh(n \lambda \log \lambda)$ time to at most $n-m+1$ instances of the \SDWC problem.
    \end{lemma}     
   
   \subsection{Combinatorial Prerequisites}
    Our improvement upon the algorithm of \cref{thm:knap} is based on \cref{fct:maxprefixes}, 
    whose analogue does not hold for \MK in general.
    Technically, instead of the notion of maximal \fr-solid prefixes, the algorithm
    relies on \emph{light \fr-solid prefixes} defined as follows:
    We say that a string $P$ of length $k$ is a light \fr-solid prefix of a weighted sequence $X$ if $k=0$ or
    $P$ is a \fr-solid prefix of $X$ such that $P[k] \ne \H(X)[k]$. 
    We symmetrically define \emph{light \fr-solid suffixes} of $X$.
    \cref{fct:maxprefixes} lets us bound the number of light solid prefixes.

    \begin{fact}\label{fct:lightprefixes}
      A weighted sequence has at most $z$ different light \fr-solid prefixes.
    \end{fact}
    \begin{proof}
      We show a pair of inverse mappings between the set of maximal \fr-solid prefixes of a weighted sequence $X$ and the
      set of light \fr-solid prefixes of $X$.
      If $P$ is a maximal \fr-solid prefix of $X$, then we obtain a light \fr-solid prefix 
      by removing all trailing letters of $P$ that are heavy letters at the corresponding positions in $X$. 
      For the inverse mapping, we extend each light \fr-solid prefix by heavy letters as long as the prefix is \fr-solid.
    \mayqed\end{proof}

    We use the notions of light \fr-solid prefixes and light \fr-solid suffixes to express a result 
    that we will use instead of \cref{lem:decomp,fct:bound}.
    
     \begin{lemma}\label{fct:key}
     Consider an instance of the \SDWC problem,
     and let $z_\ell,z_r \ge 1$ be real numbers such that $z_\ell \cdot z_r \ge z$. 
      Every consensus string $S$
      can be decomposed into $S= L \cdot c\cdot  C \cdot R$ such that the following conditions hold for some $U,V\in \{X,Y\}$:
      \begin{itemize}
        \item $L$ is a light $\frac{1}{z_\ell}$-solid prefix of $U$,
        \item $c$ is a single letter,
        \item all characters of $C$ are heavy in $V$,
        \item $R$ is a light $\frac{1}{z_r}$-solid suffix of $V$.
        \end{itemize}
    \end{lemma}
      \begin{proof}
        We set $L$ as the longest proper prefix of $S$ which is a $\frac{1}{z_{\ell}}$-solid prefix of both $X$ and $Y$,
        and we define $k := |L|$.  Note that $L$ is a light $\frac{1}{z_{\ell}}$-solid prefix of $X$ or $Y$, because $\H(X)$ and $\H(Y)$
        are dissimilar. If $k=n-1$, we conclude the proof setting $c=S[n]$ and $C=R$ to empty strings.
        
        Otherwise, we have $\P(S[1..k+1],V[1..k+1])<\frac{1}{z_\ell}$ for $V=X$ or $V=Y$.
        Since $\P(S,V)\ge \fr$ and $z_\ell \cdot z_r \ge z$, this implies
        $\P(S[k+2..n],V[k+2..n])\ge \frac{1}{z_r}$, i.e., that $C\cdot R = S[k+2..n]$ is a $\frac{1}{z_r}$-solid suffix of $V$.
        We set $C$ as the longest prefix of $S[k+2..n]$ composed of letters heavy in $V$.
        This way $R$ is clearly a light $\frac{1}{z_r}$-solid suffix of $V$.
      \end{proof}

      \subsection{Computing Light Solid Prefixes}
    
    We say that a string $P$ is a common \fr-solid prefix (suffix) of $X$ and $Y$
    if it is a \fr-solid prefix (suffix) of both $X$ and $Y$.
    A \emph{standard representation} of a common \fr-solid prefix $P$ of length $k$ of $X$ and $Y$ is a triple $(P,p_1,p_2)$
    such that $p_1$ and $p_2$ are the probabilities $p_1 = \P(P,X[1..k])$ and $p_2 = \P(P,Y[1..k])$.
    The string $P$ is written using variable-length encoding so that a letter that occurs at
    a given position with probability $p$ in $X$ or $Y$ has a representation that consists of $\Oh(\log\frac1p)$ bits. For every position $i$, the encoding can be constructed as follows: we sort letters $c$ according to $\max(\pi_i^{(X)}(c), \pi_i^{(Y)}(c))$ and assign subsequent integer identifiers according to this order.
     This lets us store a \fr-solid factor using $\Oh(\log z)$ bits: 
     we concatenate the variable-length representations of its letters and we store a bit mask of size $\Oh(\log z)$ 
     that stores the delimiters between the representations of single letters.
    An analogous representation can be applied also to common \fr-solid suffixes.
    Our assumptions on the model of computations imply that the standard representation takes constant space.
    Moreover, constant time is sufficient to extend a common \fr-solid prefix by a given letter.

      The following observation describes longer light solid prefixes in terms of shorter ones.
      \begin{observation}\label{obs:light_step}
        Let $P$ be a non-empty light \fr-solid prefix of $X$.
        If one removes its last letter and then removes all the trailing letters which are heavy at the respective
        positions in $X$, then a shorter light \fr-solid prefix of $X$ is obtained.
      \end{observation}
	We build upon \cref{obs:light_step} to derive an efficient algorithm constructing light solid prefixes.%
    \begin{lemma}\label{lem:lightprefixes_algo}
      Let $(X,Y,\fr)$ be an instance of the \SDWC problem and let $z'\le z$.
      All common \fr-solid prefixes of $X$ and $Y$ being light $\frac{1}{z'}$-solid prefixes of $X$,
      sorted first by their length and then by the probabilities in $X$,
      can be generated in $\Oh(z' (\log \log z+\log \lambda)+\log^2 z)$ time.
    \end{lemma}
      \begin{proof}
        For $k \in \{0,\ldots,n\}$, 
        let $\B_k$ be a list of the requested solid prefixes of length $k$ sorted by their probabilities $p_1$ in $X$.
        \cref{fct:lightprefixes} guarantees that $\sum_{k=0}^n |\B_k| \le z'$.
        
        We compute the lists $\B_k$ for subsequent lengths $k$.
        We start with $\B_0$ containing the empty string with its probabilities $p_1=p_2=1$.
        To compute $\B_k$ for $k>0$, we use \cref{obs:light_step}.
        We consider all candidates $i=k-1,\ldots,0$ for the length of the shorter light $\frac{1}{z'}$-solid prefix,
        and then all letters $s\ne \H(X)[k]$ to put at position $k$ of the new  light $\frac{1}{z'}$-solid prefix.
        
        For a given $i$, we iterate over all elements $(P,p_1,p_2)$ of $\B_i$ ordered by the non-increasing probabilities $p_1$, 
        and try to extend each of them by the heavy letters in $X$ at positions $i+1,\ldots,k-1$ and by the letter $s$ at position $k$.
        We process the letters $s$ ordered by $\pi_k^{(X)}(s)$, ignoring the first one ($\H(X)[k]$) and stopping as soon as we do not get a $\frac{1}{z'}$-solid prefix of $X$. 
                
        More precisely, with $X'=\H(X)$, we compute
        $$p'_1:=p_1 \cdot \prod_{j=i+1}^{k-1} \pi^{(X)}_j(X'[j]) \cdot \pi^{(X)}_k(s)\quad\mbox{and}\quad
          p'_2:=p_2 \cdot \prod_{j=i+1}^{k-1} \pi^{(Y)}_j(X'[j]) \cdot \pi^{(Y)}_k(s),$$
        check if $p'_1 \ge \frac{1}{z'}$ and $p'_2 \ge \frac1z$, and, if so, insert
        $(P \cdot X'[i+1..k-1] \cdot s,p'_1,p'_2)$ at the beginning of a new list $L_{i,s}$, indexed both by the letter $s$
        and by the length $i$ of the shorter light $\frac{1}{z'}$-solid prefix.      
        When we encounter an element $(P,p_1,p_2)$ of $\B_i$ and a letter $s$ for which $p'_1 < \frac{1}{z'}$, we proceed to the next element of $\B_i$.
        If this happens for the heaviest letter $s\ne \H(X)[k]$, we stop considering the current list $\B_i$ and proceed to $\B_{i-1}$.
        The final step consists in merging all the $k\lambda$ lists $L_{i,s}$ in the order of probabilities in $X$;
        the result is $\B_k$.

        Let us analyse the time complexity of the $k$-th step of the algorithm.
        If an element $(P,p_1,p_2)$ and letter $s$ that we consider satisfy $p'_1 \ge \frac{1}{z'}$, this accounts for a new light $\frac{1}{z'}$-solid prefix of $X$.
        Hence, in total (over all steps) we consider $\Oh(z')$ such elements.
        Note that some of these elements may be discarded due to the condition on $p'_2$.
        
        For each inspected element $(P,p_1,p_2)$, we also consider at most one letter $s$ for which $p'_1$ is not sufficiently large.
        If this is not the only letter considered for this element, such candidates can be charged to the previous class.
        The opposite situation may happen once for each list $\B_i$, which may give $\Oh(k)$ additional operations in the $k$-th step,
        $\Oh(\log^2 z)$ in total.
        
        Thanks to the order in which the lists are considered, the products of probabilities
        $\prod_{j=i+1}^{k-1} \pi^{(X)}_j(X'[j])$, $\prod_{j=i+1}^{k-1} \pi^{(Y)}_j(X'[j])$ and factors $X'[i+1..k-1]$ can be stored
        so that the representation of each subsequent light \frsq-solid prefix of length $k$ is computed in $\Oh(1)$ time.
        Finally, the merging step in the $k$-th phase takes $\Oh(|\B_k|\log(k\lambda)) = \Oh(|\B_k| (\log \log z+\log \lambda))$ time
        if a binary heap is used.

        The time complexity of the whole algorithm is $\Oh(\log^2 z + \sum_{k=1}^{n}|\B_k| (\log \log z+\log \lambda))$.
        By the already mentioned \cref{fct:lightprefixes}, this is $\Oh(\log^2 z+z' (\log \log z+\log \lambda))$.
      \mayqed\end{proof}

    \subsection{Merge-in-the-Middle Implementation}
    In this section we apply \cref{fct:key} to solve the \SDWC problem.
    We use \cref{lem:lightprefixes_algo} to generate all candidates for $L\cdot c$ and $R$,
    and we apply a divide-and-conquer procedure to fill this with $C$.
    Our procedure works for fixed $U,V\in \{X,Y\}$; the algorithm repeats it for all four choices.
    
      Let $\L_i$ denote a list of all common \fr-solid prefixes of $X$ and $Y$ obtained by extending
      a light $\frac{\sqrt{\lambda}}{\sqrt{z}}$-solid prefix of $U$ of length $i-1$ by a single letter $s$ at position $i$,
      and let $\R_i$ denote a list of all common $\frac{1}{z}$-solid suffixes of $X$ and $Y$ of length $n-i+1$ that are light $\frac1{\sqrt{z\lambda}}$-solid suffixes of $V$.  We assume that the lists $\L_i$ and $\R_i$ are sorted according to the probabilities in $U$ and $V$, respectively.

   \begin{lemma}\label{lem:L_R}
      The lists $\L_i$ and $\R_i$ for $i \in \{1,\ldots,n+1\}$ can be computed in $\Oh(\sqrt{z\lambda} (\log \log z+\log \lambda))$ time.
      Their total size is $\Oh(\sqrt{z \lambda})$.
    \end{lemma}
      \begin{proof}
        $\Oh(\sqrt{z\lambda} (\log \log z+\log\lambda))$-time computation of the lists $\R_i$ is directly due to \cref{lem:lightprefixes_algo}.
        As for the lists $\L_i$, we first compute in $\Oh(\frac{\sqrt{z}}{\sqrt{\lambda}}(\log \log z+\log\lambda))$ time the lists of all light $\frac{\sqrt{\lambda}}{\sqrt{z}}$-solid prefixes of $U$, sorted by the lengths of strings and then by the probabilities in $U$, again using \cref{lem:lightprefixes_algo}.
        Then for each length $i-1$ and for each letter $s$ at the $i$-th position, we extend all these prefixes by a single letter.
        This way we obtain $\lambda$ lists for a given $i-1$ that can be merged according to the probabilities in $U$ to form the list~$\L_i$.
        Generation of the auxiliary lists takes $\Oh(\frac{\sqrt{z}}{\sqrt{\lambda}}\cdot \lambda)=\Oh(\sqrt{z\lambda})$ time in total,
        and merging them using a binary heap takes $\Oh(\sqrt{z\lambda} \log \lambda)$ time.
        This way we obtain an $\Oh(\sqrt{z\lambda} (\log \log z+\log\lambda))$-time algorithm.
      \mayqed\end{proof}

      Let $\Ls_{a,b}$ be a list of common \fr-solid prefixes of $X$ and $Y$ of length $b$
      obtained by taking a common \fr-solid prefix from $\L_i$ for some $i \in \{a,\ldots,b\}$
      and extending it by $b-i$ letters that are heavy at the respective positions in $V$.
      Similarly, $\Rs_{a,b}$ is a list of common \fr-solid suffixes of length $n-a+1$
      obtained by taking a common \fr-solid suffix from $\R_i$ for some $i \in \{a,\ldots,b\}$
      and prepending it by $i-a$ letters that are heavy in $V$.
      Again, we assume that each of the lists $\Ls_{a,b}$ and $\Rs_{a,b}$ is sorted according to the probabilities in $U$ and $V$, respectively.
      
      A \emph{basic interval} is an interval $[a,b]$ represented by its endpoints $1 \le a \le b \le n+1$ such that
      $2^j \mid a-1$ and $b=\min(n+1,a+2^j-1)$ for some integer $j$ called the \emph{layer} of the interval.
      For every $j=0,\ldots,\ceil{\log n}$, there are $\Theta(\frac{n}{2^j})$ basic intervals and they are pairwise disjoint.
      
      \begin{example}
        For $n=7$, the basic intervals are $[1,1],\ldots,[8,8],[1,2],[3,4],[5,6],[7,8],\allowbreak[1,4],[5,8],[1,8]$.
      \end{example}

      \begin{lemma}\label{lem:Ls_Rs}
        The lists $\Ls_{a,b}$ and $\Rs_{a,b}$ for all basic intervals $[a,b]$ use $\Oh(\sqrt{z\lambda}\log\log z)$ space
        and can be constructed in $\Oh(\sqrt{z\lambda}(\log\log z+\log \lambda))$ time.
      \end{lemma}
      \begin{proof}
        We compute all the lists $\Ls_{a,b}$ and $\Rs_{a,b}$ for consecutive layers $j=0,\ldots,\ceil{\log n}$ of basic intervals $[a,b]$.
        For $j=0$, we have $\Ls_{a,a} = \L_a$ and $\Rs_{a,a} = \R_a$.
        Suppose that we wish to compute $\Ls_{a,b}$ for $a<b$ at layer $j$ (the computation of $\Rs_{a,b}$ is symmetric).
        Take $c=a+2^{j-1}-1$.
        Let us iterate through all the elements $(P,p_1,p_2)$ of the list $\Ls_{a,c}$, extend each string $P$ by $\H(V)[c+1..b]$,
        and multiply the probabilities $p_1$ and $p_2$ by
        $$\prod_{i=c+1}^{b} \pi^{(X)}_i(\H(V)[i]) \quad\mbox{and}\quad \prod_{i=c+1}^{b} \pi^{(Y)}_i(\H(V)[i]),$$
        respectively.
        If a common \fr-solid prefix is obtained, it is inserted at the end of an auxiliary list $L$.
        The resulting list $L$ is merged with $\Ls_{c+1,b}$ according to the probabilities in $U$; the result is $\Ls_{a,b}$.

        Thus, we can compute $\Ls_{a,b}$ in time proportional to the sum of lengths of $\Ls_{a,c}$ and $\Ls_{c+1,b}$.
        (Note that the necessary products of probabilities can be computed in $\Oh(n) = \Oh(\log z)$ total time.)
        For every $j=1,\ldots,\ceil{\log n}$, the total length of the lists from the $j$-th layer
        does not exceed the total length of the lists from the $(j-1)$-th layer.
        By \cref{lem:L_R}, the lists at the $0$-th layer have size $\Oh(\sqrt{z\lambda})$.
        The conclusion follows from the fact that $\log n = \Oh(\log\log z)$.
      \mayqed\end{proof}

 	  Next, we provide an analogue of \cref{lem:knap2}.
      \begin{lemma}\label{lem:meet}
        Let $L$ and $R$ be lists containing, for some $k\in\{0,\ldots,n\}$, standard representations of common \fr-solid prefixes of length $k$ and
        common \fr-solid suffixes of length $n-k$ of $X$ and $Y$.
        If the elements of each list are sorted according to non-decreasing probabilities in $X$ or $Y$,
        one can check in $\Oh(|L|+|R|)$ time whether the concatenation of any \fr-solid prefix from $L$
        and \fr-solid suffix from $R$ yields a string $S$ such that $S \match X$ and $S \match Y$.
      \end{lemma}
      \begin{proof}
        First, we filter out dominated elements of the lists, i.e., elements $(P,p_1,p_2)$ such that there exists
        another element $(P',p_1',p_2')$ with $p_1'> p_1$ and $p_2'> p_2$. 
        After this operation,  we make sure that both lists are sorted with respect to the non-decreasing probabilities in $X$;
        this might require reversing the list.

        For every element $(P,p_1,p_2)$ of $L$, we compute the first (leftmost) element $(P',p'_1,p'_2)$ of $R$
        such that $p_1 p'_1 \ge \frac1z$. This element maximizes $p'_2$ among all elements satisfying the latter condition.
        Hence, it suffices to check if $p_2 p'_2 \ge \frac1z$, and if so, report the result $S=PP'$.
        As the lists are ordered by $p_1$ and $p'_1$, respectively, all such elements
        can be computed in $\Oh(|L|+|R|)$ total time.
      \mayqed\end{proof}

      Finally, we are ready to apply a divide-and-conquer approach to the \SDWC problem:

      \begin{lemma}\label{lem:DWM_hard}
        The \SDWC problem can be solved in $\Oh(\sqrt{z\lambda} (\log \log z + \log \lambda))$ time.
      \end{lemma}
      \begin{proof}
        The algorithm goes along \cref{fct:key}, considering all choices of $U$ and $V$.
        For each of them, we proceed as follows:
        
        First, we compute the lists $\L_i$, $\R_i$ and $\Ls_{a,b}$, $\Rs_{a,b}$ for all basic intervals.
        By \cref{lem:L_R,lem:Ls_Rs}, this takes $\Oh(\sqrt{z\lambda} (\log\log z+\log \lambda))$ time.
        In order to find out if there is a feasible solution, it suffices to attempt joining
        a common \fr-solid prefix from $\L_j$ with a common \fr-solid suffix from $\R_k$ for some indices $1 \le j < k \le n+1$
        by heavy letters of $V$ at positions $j+1,\ldots,k-1$.

        We use a recursive routine to find such a pair of indices $j$, $k$ in a basic interval $[a,b]$
        which has positive length and therefore can be decomposed into two basic subintervals $[a,c]$ and $[c+1,b]$.
        Then either $j \le c < k$, or both indices $j$, $k$ belong to the same interval $[a,c]$ or $[c+1,b]$.
        To check the former case, we apply the algorithm of \cref{lem:meet} to $L = \Ls_{a,c}$ and $R = \Rs_{c+1,b}$.
        The two latter cases are solved by recursive calls for the subintervals.

        The recursive routine is called first for the basic interval $[1,n+1]$.
        The computations performed by the routine for the basic intervals at the $j$-th level
        take at most the time proportional to the total size of lists $\Ls_{a,b}$, $\Rs_{a,b}$ at the $(j-1)$-th level.
        \cref{lem:Ls_Rs} shows that the total size of the lists at all levels is $\Oh(\sqrt{z\lambda} \log\log z)$.
        Consequently, the whole procedure works in $\Oh(\sqrt{z\lambda} (\log\log z+\log \lambda))$ time.
      \mayqed\end{proof}
      \cref{lem:DWM_hard} combined with \cref{lem:sdwc} provides an efficient implementation of the \GWPMFull.

    \begin{theorem}\label{lem:gwpm}
    The \GWPM problem can be solved in $\Oh(n\sqrt{z\lambda}(\log \log z + \log \lambda))$ time.
     An oracle for the \GWPM problem using $\Oh(n \log z)$ space and supporting queries in $\Oh(m)$ time can be
        computed within the same time complexity.
    \end{theorem}

    \section{Faster Algorithms for Large $\lambda$}\label{app:fast}
    In this section we analyse the running times of algorithms for the \MK problem expressed
    as $\Oh(n^{\Oh(1)}\cdot T(a,\lambda))$  for some function $T$ monotone with respect to both arguments.
    The algorithm of \cref{thm:knap} proves that achieving $T(a,\lambda)=\sqrt{a\lambda}$ is possible. 
    On the other hand, if we assume that \SubsetSum does not admit an $\Oh^*(2^{(0.5-\varepsilon)n})$-time solution,
    then we immediately get that we cannot have $T(a,2)=\Oh(a^{0.5 -\varepsilon})$ for any $\varepsilon \ge 0$.
    Similarly, the 3-\Sum conjecture implies that $T(\lambda^3,\lambda)=\Oh(\lambda^{2-\varepsilon})$ is impossible.
    While this already refutes the possibility of having $T(a,\lambda)=\Oh(a^{0.5}\lambda^{0.5-\varepsilon})$
    across all arguments $(a,\lambda)$, such a bound may still hold for some special cases covering an infinite number of arguments.
    For example, we may potentially achieve $T(a,\lambda)=\Oh((a\lambda)^{0.5-\varepsilon})=\Oh(\lambda^{1.5-\varepsilon})$ for $a=\lambda^2$.
    
    Before we prove that this is indeed possible, let us see the consequences of the hardness of 3-\Sum and, in general, $(2k-1)$-\Sum.   
    For a non-negative integer $k$, the $(2k-1)$-\Sum conjecture refutes $T(\lambda^{2k-1},\lambda)=\Oh(\lambda^{k-\varepsilon})$.
    By monotonicity of $T$ with respect to the first argument, we conclude that $T(\lambda^{c},\lambda)=\Oh(\lambda^{k-\varepsilon})$
    is impossible for $c\ge 2k-1$. 
    On the other hand, monotonicity with respect to the second argument shows that
    $T(\lambda^{c},\lambda)=\Oh(\lambda^{c\frac{k}{2k-1}-\varepsilon})$
    is impossible for $c\le 2k-1$. The lower bounds following from $(2k-1)$-\Sum and $(2k+1)$-\Sum
    turn out to meet at $c=2k-1+\frac{1}{k+1}$; see \cref{fig:graph}.

       \begin{figure}[hb]
   \begin{center}
    \begin{tikzpicture}
  \draw[->] (0.5,0.5) -- (7.5,0.5) node[right] {$c$};
  \draw[->] (0.5,0.5) -- (0.5,4.5);
  \draw (1,1) -- (1.5, 1) -- (3,2) -- (10/3,2) -- (5,3) -- (5.25,3) -- (7,4);
  \draw[thin, dotted] (1,1) -- (7,4);
  \foreach \x in {1,...,7} {
  	\draw (\x,0.6) -- (\x, .4) node[below] {\tiny $\x$};
  }
  \foreach \x in {1,...,4} {
  	\draw (0.6,\x) -- (.4,\x) node[left] {\tiny $\x$};
  }
  
\end{tikzpicture}
\end{center}
\caption{Illustration of the upper bound (dotted) and lower bound (solid) on $\log_{\lambda}T(\lambda^c,\lambda)$.}\label{fig:graph}
\end{figure}
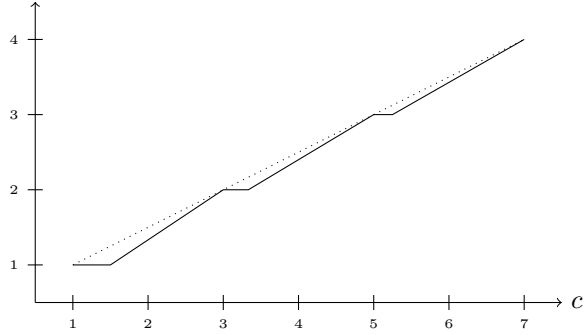

  Consequently, we have some room between the lower and the upper bound of $\sqrt{a \lambda}$.
	In the aforementioned case of $a=\lambda^2$, the upper bound is $\lambda^{\frac32}$,
	compared to the lower bound of $\lambda^{\frac43-\varepsilon}$.
	Below, we show that the upper bound can be improved to meet the lower bound.
	More precisely, we show an algorithm whose running time is $\Oh(N + (a^{\frac{k+1}{2k+1}}+\lambda^k)\log\lambda \cdot n^k)$
	for every positive integer $k$.
	Note that $a^{\frac{k+1}{2k+1}}+\lambda^k = \lambda^{c\frac{k+1}{2k+1}}+ \lambda^k$, so for $2k-1\le c \le 2k+1$ the running time indeed matches the lower bounds up to the $n^k$ term.

    Due to \cref{lem:knapred2}, the extra $n^k$ term reduces to $\Oh((\frac{\log A}{\log \lambda})^k)$,
    and if we measure the running time using $A$ instead of $a$, it becomes a constant ($k^{\Oh(k)}$).
    In particular, this lets us prove that the \GWPM problem can be solved in $\Oh(n(z^{\frac{k+1}{2k+1}}+\lambda^k)\log\lambda)$ time
    for any integer $k=\Oh(1)$, improving upon the solution of \cref{sec:GWPMReduction}
    unless $z=\lambda^{\omega(1)}$ or $z=\lambda^{c\pm o(1)}$ for an odd integer $c$.

    \subsection{Algorithm for Multichoice Knapsack}\label{app:fastmk}    
	Let us start by discussing the bottleneck of the algorithm of \cref{thm:knap} for large $\lambda$.
	The problem is that the size of the classes does not let us partition every choice $S$ into a prefix $L$ and a suffix $R$
	with ranks both $\Oh(\sqrt{A_V})$. \cref{lem:decomp} leaves us with an extra letter $c$ between $L$ and $R$,
	and in the algorithm we append it to the prefix (while generating $\L_{j-1}^{(\ell)}\odot C_j$).

    We provide a workaround based on reordering of classes.
    Our goal is to make sure that items with large rank appear only in a few leftmost classes.
    For this, we guess the classes of the $k$ items with largest rank (in a feasible solution) and move them to the front.
    Since this depends on the sought feasible solution, we shall actually verify all $\binom{n}{k}$ possibilities.
    
    Now, our solution considers two cases:
    For $j>k$, the reordering lets us assume $\rank_V(c)\le \ell^{\frac{1}{k}}$, so we do not need
    to consider all items from $C_j$.
    For $j\le k$, on the other hand, we exploit the fact that $|\L_{j-1}^{(\ell)}\odot C_j|\le \lambda^{j}$,
    which at most $\lambda^k$.
         
    Combinatorial foundation of this intuition is formalized as a variant of \cref{lem:decomp}:
     
  \begin{lemma}\label{lem:decomp2}
  Let $\ell$ and $r$ be positive integers such that $V^{(\ell)}_{\L_j}+V^{(r)}_{\R_{j+1}}> V$  for every $0\le j \le n$.
  Let $k \in \{1,\ldots,n\}$ and suppose that $S$ is a choice with $v(S)\le V$ such that $\rank_V(S\cap C_i)\ge \rank_V(S\cap C_j)$ for $i \le k < j$.
    There is an index $j\in\{1,\ldots,n\}$ and a decomposition $S=L\cup\{c\}\cup R$
  such that $L\in \L_{j-1}^{(\ell)}$, $R\in \R_{j+1}^{(r)}$, $c\in C_j$, and either $\rank_V(c)\le \ell^{\frac{1}{k}}$ or $j \le k$.
  \end{lemma}
  \begin{proof}
  We claim that the decomposition constructed in the proof of \cref{lem:decomp} satisfies the extra $\rank_V(c)\le \ell^{\frac{1}{k}}$
  condition if $j > k$. Let $S = \{c_1,\ldots,c_n\}$ and $S_i = \{c_1,\ldots,c_i\}$.
  Obviously $\rank_V(c_i)\ge 1$ for $k < i < j$ and, by the extra assumption, $\rank_V(c_i)\ge \rank_V(c)$ for $1\le i \le k$. Hence, \cref{fct:comb} yields $\rank_V(S_{j-1})\ge \rank_V(c)^k$. Simultaneously, we have $v(S_{j-1})<V^{(\ell)}_{\L_j}$,
  so $\rank_V(S_{j-1})<\ell$. Combining these inequalities, we immediately get the claimed bound.
  \end{proof}

   \begin{theorem}\label{thm:knap3}
    For every positive integer $k=\Oh(1)$, the \MK problem can be solved in $\Oh(N+\allowbreak {(a^{\frac{k+1}{2k+1}}+\lambda^k)}\log A (\frac{\log A}{\log \lambda})^{k})$ time.
  \end{theorem}
  \begin{proof}
  As in the proof of \cref{thm:knap}, we actually provide an algorithm whose running time depends on $A_V$ rather than $a$.
  Moreover, \cref{lem:knapred2} lets us assume that $n=\Oh(\frac{\log A}{\log \lambda})$.   
    
    We first guess the $k$ positions where items with largest ranks $\rank_V$ are present in the solution $S$ and move these positions to the front. This gives $\binom{n}{k}=\Oh((\frac{\log A}{\log \lambda})^k)$ possible selections. For each of them, we proceed as follows.

    We increment an integer $r$ starting from $1$, maintaining $\ell=\big\lceil r^{\frac{k}{k+1}}\big\rceil$ and all the lists $\L_{j}^{(\ell)}$ and $\R_{j+1}^{(r)}$ for $0\le j \le n$,
    as long as $V^{(\ell)}_{\L_j}+V^{(r)}_{\R_{j+1}}\le V$ for some $j$.
    By \cref{fct:bound}, we stop with $r=\Oh(A_V^{\frac{k+1}{2k+1}})$ and
    thus the total time of this phase is $\Oh(A_V^{\frac{k+1}{2k+1}}\log A)$ due to the online procedure of \cref{lem:generate}.

    By \cref{lem:decomp2}, every feasible solution $S$  admits a decomposition $S=L\cup\{c\}\cup R$ for some $j$;
    we shall consider all possibilities for $j$.
    For each of them, we reduce searching for $S$ to an instance of the \MK  problem with $N'=\Oh(A_V^{\frac{k+1}{2k+1}}+\lambda^k)$ and $n'=2$. 
    By \cref{lem:knap2}, these instances can be solved in $\Oh((A_V^{\frac{k+1}{2k+1}}+\lambda^k)\frac{\log A}{\log \lambda})$
    time in total.
   
    For $j\le k$, the items of the $j$-th instance are going to belong to classes $\L_{j-1}^{(\ell)}\odot C_j$ and $\R_{j+1}^{(r)}$,
    where $\L_{j-1}^{(\ell)}\odot C_j = \{L\cup \{c\} : L\in \L_{j-1}^{(\ell)} , c\in C_j\}$.
    The set $\L_{j-1}^{(\ell)}\odot C_j$ can be sorted by merging $|C_j|$ sorted lists of size at most $\lambda^{j-1}$ each,
    i.e., in $\Oh(\lambda^k \log \lambda)$ time.
    On the other hand, for $j > k$, we take 
    $\{L\cup \{c\} : L\in \L_{j-1}^{(\ell)} , c\in C_j, \rank_V(c)\le \ell^{\frac{1}{k}}\}$ and $\R_{j+1}^{(r)}$.
    The former set can be constructed by merging $\ell^{\frac{1}{k}}=\Oh(r^\frac{1}{k+1})$ sorted lists of size $\Oh(r^\frac{k}{k+1})$ each,
    i.e., in $\Oh(r\log \lambda)=\Oh(A_V^{\frac{k+1}{2k+1}}\log \lambda)$ time. 
    
    Summing up over all indices $j$, this gives $\Oh((A_V^{\frac{k+1}{2k+1}} + \lambda^k)\log A)$ time
    for a single selection of the $k$ positions with largest ranks,
    and $\Oh((A_V^{\frac{k+1}{2k+1}} + \lambda^k)\log A (\frac{\log A}{\log \lambda})^{k})$ in total.
       
    Clearly, each solution of the constructed instances represents a solution of the initial instance,
    and by \cref{lem:decomp2}, every feasible solution of the initial instance has its counterpart in one of the constructed instances. 
  
      Before we conclude the proof, let us note that the optimal $k$ does not need to be known in advance. 
    To deal with this issue, we try consecutive integers $k$ and stop the procedure if
    \cref{fct:bound} yields that $A_V > \lambda^{2k+1}$, i.e., if $r$ is incremented beyond $\lambda^{k+1}$.
    If the same happens for the other instance of the algorithm (operating on $\rank_W$ instead of $\rank_V$), we conclude that $a>\lambda^{2k+1}$,
    and thus we shall better use larger $k$. 
    The running time until this point is $\Oh(\lambda^{k+1}\log\lambda (\frac{\log A}{\log \lambda})^k)$ due to \cref{lem:generate}. 
    On the other hand, if $r\le \lambda^{k+1}$, the algorithm behaves as if $a \le  \lambda^{2k+1}$, i.e., runs in $\Oh(\lambda^{k+1}\log\lambda (\frac{\log A}{\log \lambda})^k)$ time.
     This workaround (considering all smaller values $k$) adds extra $\Oh(\lambda^{k}\log\lambda (\frac{\log A}{\log \lambda})^{k-1})$
    to the time complexity for the optimal value $k$, which is less than the upper bound on the running time we have for this value $k$.
      \end{proof}

    If we are to bound the complexity in terms of $A$ only, the running time becomes
    $${\Oh(N+ {(A^{\frac{k+1}{2k+1}}+\lambda^k)}\log A (\frac{\log A}{\log \lambda})^{k})}.$$
    Assumptions that $A\le \lambda^{2k+1}$ and $k=\Oh(1)$ let us get rid of the $(\frac{\log A}{\log \lambda})^{k}$ term, which can be bounded by $(2k+1)^k=\Oh(1)$.
    If one of these assumptions is not satisfied, we can improve the bound on the running time anyway:
    using \cref{thm:knap3} with increased $k$ if $A> \lambda^{2k+1}$, and using \cref{thm:knap} if $k=\omega(1)$.
    
    \begin{corollary}
    Let $k=\Oh(1)$ be a positive integer such that $A\le \lambda^{2k+1}$. The \MK problem can be solved in $\Oh(N+ {(A^{\frac{k+1}{2k+1}}+\lambda^k)}\log \lambda)$ time.
    \end{corollary}
    
    This leads to the following results for weighted pattern matching:
    
   \begin{theorem}\label{thm:fast}
	Suppose that $\lambda^{2k-1}\le z \le \lambda^{2k+1}$ for some positive integer $k=\Oh(1)$.
	Then the \SDWC problem can be solved in $\Oh((z^{\frac{k+1}{2k+1}} + \lambda^{k})\log\lambda)$ time,
	and the \GWPM problem can be solved in $\Oh(n(z^{\frac{k+1}{2k+1}} + \lambda^{k})\log\lambda)$ time.
	\end{theorem}	
    
    As we noted at the beginning of this section, \cref{lem:red} implies that
    any improvement of the dependence of the running time on $z$ or $\lambda$ by $z^{\varepsilon}$ (equivalently, by $\lambda^{\varepsilon}$)
    would contradict the $k$-\Sum conjecture.

  \section{Final Remarks}\label{sec:fr}
  In \cref{sec:MK}, we gave an $\Oh(N+a^{0.5}\lambda^{0.5}\log A)$-time algorithm for the \MK problem.
  Improvement of either exponent to $0.5 - \varepsilon$ would result in a breakthrough for the \SubsetSum and 3-\Sum problems,
  respectively. 
  Nevertheless, this does not refute the existence of faster algorithms for some particular values $(a,\lambda)$
  other than those emerging from instances of \SubsetSum or 3-\Sum.
  In \cref{app:fast}, we show an algorithm that is superior if $\frac{\log a}{\log \lambda}$ is a constant other than an odd integer.
  We also prove it to be optimal (up to lower order terms) for every constant $\frac{\log a}{\log \lambda}$
  unless the $k$-\Sum conjecture fails.

  \bibliographystyle{plain}
  \bibliography{wpm}

\begin{thebibliography}{10}

\bibitem{amir_weighted_property_matching_j}
Amihood Amir, Eran Chencinski, Costas~S. Iliopoulos, Tsvi Kopelowitz, and Hui
  Zhang.
\newblock Property matching and weighted matching.
\newblock {\em Theor. Comput. Sci.}, 395(2-3):298--310, April 2008.

\bibitem{CPM2016}
Carl Barton, Tomasz Kociumaka, Jakub Radoszewski, and Solon~P. Pissis.
\newblock Efficient index for weighted sequences, 2016.
\newblock Accepted to Combinatorial Pattern Matching, {CPM} 2016.

\bibitem{DBLP:conf/cwords/BartonP15}
Carl Barton and Solon~P. Pissis.
\newblock Linear-time computation of prefix table for weighted strings.
\newblock In Florin Manea and Dirk Nowotka, editors, {\em Combinatorics on
  Words, {WORDS} 2015}, volume 9304 of {\em LNCS}, pages 73--84. Springer,
  2015.

\bibitem{DBLP:conf/edbt/BiswasPTS16}
Sudip Biswas, Manish Patil, Sharma~V. Thankachan, and Rahul Shah.
\newblock Probabilistic threshold indexing for uncertain strings.
\newblock In Evaggelia Pitoura, Sofian Maabout, Georgia Koutrika, Am{\'{e}}lie
  Marian, Letizia Tanca, Ioana Manolescu, and Kostas Stefanidis, editors, {\em
  19th International Conference on Extending Database Technology, {EDBT} 2016},
  pages 401--412. OpenProceedings.org, 2016.

\bibitem{KCL_publication}
Manolis Christodoulakis, Costas~S. Iliopoulos, Laurent Mouchard, and Kostas
  Tsichlas.
\newblock Pattern matching on weighted sequences.
\newblock In {\em Algorithms and Computational Methods for Biochemical and
  Evolutionary Networks, CompBioNets 2004}, KCL publications, 2004.

\bibitem{AlgorithmsOnStrings}
Maxime Crochemore, Christophe Hancart, and Thierry Lecroq.
\newblock {\em Algorithms on Strings}.
\newblock Cambridge University Press, New York, NY, USA, 2007.

\bibitem{DBLP:conf/mfcs/EtscheidKMR15}
Michael Etscheid, Stefan Kratsch, Matthias Mnich, and Heiko R{\"{o}}glin.
\newblock Polynomial kernels for weighted problems.
\newblock In Giuseppe~F. Italiano, Giovanni Pighizzini, and Donald Sannella,
  editors, {\em Mathematical Foundations of Computer Science, {MFCS} 2015, Part
  {II}}, volume 9235 of {\em LNCS}, pages 287--298. Springer, 2015.

\bibitem{DBLP:journals/jacm/FredmanKS84}
Michael~L. Fredman, J{\'a}nos Koml{\'o}s, and Endre Szemer{\'e}di.
\newblock Storing a sparse table with {$O(1)$} worst case access time.
\newblock {\em J. {ACM}}, 31(3):538--544, 1984.

\bibitem{DBLP:journals/comgeo/GajentaanO95}
Anka Gajentaan and Mark~H. Overmars.
\newblock On a class of {$O(n^2)$} problems in computational geometry.
\newblock {\em Comput. Geom.}, 5:165--185, 1995.

\bibitem{DBLP:books/daglib/0069796}
Eitan~M. Gurari.
\newblock {\em Introduction to the theory of computation}.
\newblock Computer Science Press, 1989.

\bibitem{DBLP:journals/jacm/HorowitzS74}
Ellis Horowitz and Sartaj Sahni.
\newblock Computing partitions with applications to the knapsack problem.
\newblock {\em J. {ACM}}, 21(2):277--292, 1974.

\bibitem{costas_weighted_suffix_tree_j}
Costas~S. Iliopoulos, Christos Makris, Yannis Panagis, Katerina Perdikuri,
  Evangelos Theodoridis, and Athanasios~K. Tsakalidis.
\newblock The weighted suffix tree: An efficient data structure for handling
  molecular weighted sequences and its applications.
\newblock {\em Fundam. Inform.}, 71(2-3):259--277, 2006.

\bibitem{DBLP:journals/jcss/ImpagliazzoP01}
Russell Impagliazzo and Ramamohan Paturi.
\newblock On the complexity of $k$-{SAT}.
\newblock {\em J. Comput. Syst. Sci.}, 62(2):367--375, 2001.

\bibitem{DBLP:books/daglib/0010031}
Hans Kellerer, Ulrich Pferschy, and David Pisinger.
\newblock {\em Knapsack problems}.
\newblock Springer, 2004.

\bibitem{ETHsurvey}
Daniel Lokshtanov, D{\'{a}}niel Marx, and Saket Saurabh.
\newblock Lower bounds based on the {E}xponential {T}ime {H}ypothesis.
\newblock {\em Bulletin of the {EATCS}}, 105:41--72, 2011.

\bibitem{DBLP:journals/tcs/PizziU08}
Cinzia Pizzi and Esko Ukkonen.
\newblock Fast profile matching algorithms - {A} survey.
\newblock {\em Theor. Comput. Sci.}, 395(2-3):137--157, 2008.

\bibitem{DBLP:journals/jcb/RajasekaranJS02}
Sanguthevar Rajasekaran, X.~Jin, and John~L. Spouge.
\newblock The efficient computation of position-specific match scores with the
  fast {F}ourier transform.
\newblock {\em J. Comp. Biol.}, 9(1):23--33, 2002.

\bibitem{DBLP:conf/icalp/Ruzic08}
Milan Ruzic.
\newblock Constructing efficient dictionaries in close to sorting time.
\newblock In Luca Aceto, Ivan Damg{\aa}rd, Leslie~Ann Goldberg, Magn{\'{u}}s~M.
  Halld{\'{o}}rsson, Anna Ing{\'{o}}lfsd{\'{o}}ttir, and Igor Walukiewicz,
  editors, {\em Automata, Languages and Programming, {ICALP} 2008, Part {I}},
  volume 5125 of {\em LNCS}, pages 84--95. Springer, 2008.

\bibitem{JASPAR}
Albin Sandelin, Wynand Alkema, P\"ar Engstr\"om, Wyeth~W. Wasserman, and Boris
  Lenhard.
\newblock {JASPAR}: an open-access database for eukaryotic transcription factor
  binding profiles.
\newblock {\em Nucl. Acids Res.}, 32(1):D91--D94, 2004.

\end{thebibliography}

\end{document}